\documentclass{amsart}
\usepackage[left=1in, right=1in, top=1in, bottom=1in]{geometry}
\usepackage{mathtools}
\usepackage{enumitem}
\usepackage{graphicx}
\usepackage{amsmath, amsthm, amssymb}
\usepackage{xspace}
\usepackage{comment}
\usepackage[hidelinks]{hyperref}
\usepackage[dvipsnames]{xcolor}
\usepackage{algorithm}
\usepackage{algpseudocode}
\usepackage{makecell}
\usepackage{array}
\usepackage[dvipsnames]{xcolor}

\usepackage{tikz}
\usetikzlibrary{shapes, calc, decorations.pathreplacing, calligraphy, arrows.meta}
\usepackage{subcaption}
\tikzstyle{myStyle}=[shape = circle, minimum size = 2pt, inner sep =1.5pt, outer sep = 0pt, draw, fill=white]

\newcommand{\ab}[1]{\left|#1\right|} 

\newcommand{\algprobm}[1]{\textsc{#1}\xspace}

\newcommand{\stg}{\operatorname{\#ST}}
\newcommand{\mps}{\operatorname{\#MPS}}
\newcommand{\setsum}{\operatorname{sum}}
\newcommand{\compset}{\mathcal{B}'}

\theoremstyle{plain}
\newtheorem{theorem}{Theorem}[section]
\newtheorem{proposition}[theorem]{Proposition}
\newtheorem{corollary}[theorem]{Corollary}
\newtheorem{lemma}[theorem]{Lemma}

\newtheorem{observation}[theorem]{Observation}

\newtheorem*{thm:MainCounting}{Theorem~\ref{thm:MainCounting}}
\newtheorem*{thm:MainMedian}{Theorem~\ref{thm:MainMedian}}

\theoremstyle{definition}
\newtheorem{definition}[theorem]{Definition}
\newtheorem{remark}[theorem]{Remark}
\newtheorem{question}[theorem]{Question}

\newtheorem*{thm:MainFPT}{Theorem~\ref{thm:MainFPT}}

\newcommand{\p}[1]{(#1)}
\newcommand{\pb}[1]{\left(#1\right)}
\newcommand{\st}[1]{\left\{#1\right\}}

\DeclareMathOperator{\poly}{poly}

\title{Pairwise Rearrangement is Fixed-Parameter Tractable in the Single Cut-and-Join Model}

\thanks{We wish to thank the American Mathematical Society for organizing the Mathematics Research Community workshop where this work began. We also wish to thank the anonymous referees for their helpful feedback. This material is based upon work supported by the National Science Foundation under Grant Number DMS 1641020. A preliminary version appeared in the proceedings of SWAT 2024 \cite{BaileySWAT}.} 

\author[Bailey et. al.]{Lora Bailey}
\address[Bailey]{Department of Mathematics, Grand Valley State University, Allendale, MI}
\email{baileylo@gvsu.edu}
\author[]{Heather Smith Blake}
\address[Blake]{Department of Mathematics and Computer Science, Davidson College, Davidson, NC}
\email{hsblake@davidson.edu}
\author[]{Garner Cochran}
\address[Cochran]{Department of Mathematics and Computer Science, Berry College, Mount Berry, GA}
\email{gcochran@berry.edu}
\author[]{Nathan Fox}
\address[Fox]{Department of Quantitative Sciences, Canisius University, Buffalo, NY}
\email{fox42@canisius.edu}
\author[]{Michael Levet$^*$}
\address[Levet]{Department of Computer Science, College of Charleston, Charleston, SC}
\email{levetm@cofc.edu}
\author[]{Reem Mahmoud}
\address[Mahmoud]{Division of Science, NYU Abu Dhabi, Abu Dhabi, UAE  (\textup{The majority of the work was completed while at} Department of Computer Science, Virginia Commonwealth University, Richmond, VA)}
\email{rm7230@nyu.edu}
\author[]{Inne Singgih}
\address[Singgih]{Department of Mathematical Sciences, University of Cincinnati, Cincinnati, OH}
\email{inne.singgih@uc.edu}
\author[]{Grace Stadnyk}
\address[Stadnyk]{Department of Mathematics, Furman University, Greenville, SC}
\email{grace.stadnyk@furman.edu}
\author[]{Alexander Wiedemann}
\address[Wiedemann]{Department of Mathematics and Computational Data Science, Hamline University, Saint Paul, MN  (\textup{The majority of the work was completed while at} Department of Mathematics, Randolph-Macon College, Ashland, VA)}
\email{awiedemann01@hamline.edu}

\begin{document}
\let\thefootnote\relax\footnotetext{$^*$Corresponding author}

\begin{abstract}
Genome rearrangement is a common model for molecular evolution. In this paper, we consider the \algprobm{Pairwise Rearrangement} problem, which takes as input two genomes and asks for the number of minimum-length sequences of permissible operations transforming the first genome into the second. In the Single Cut-and-Join model (Bergeron, Medvedev, \& Stoye, \textit{J. Comput. Biol.} 2010), \algprobm{Pairwise Rearrangement} is $\#\textsf{P}$-complete (Bailey, et. al., COCOON 2023; \textit{Theor. Comput. Sci.} 2024), which implies that exact sampling is intractable. In order to cope with this intractability, we investigate the parameterized complexity of this problem. We exhibit a fixed-parameter tractable algorithm with respect to the number of components in the \emph{adjacency graph} that are not cycles of length $2$ or paths of length $1$. As a consequence, we obtain that \algprobm{Pairwise Rearrangement} in the Single Cut-and-Join model is fixed-parameter tractable by distance. Our results suggest that the number of nontrivial components in the adjacency graph serves as the key obstacle for efficient sampling.
\end{abstract}

\maketitle

\section{Introduction}

With the natural occurrence of mutations in genomes and the wide range of effects this can incite, scientists seek to understand the evolutionary relationship between species. Several discrete mathematical models have been proposed to model these mutations based on biological observations. Genome rearrangement models consider situations in which large-scale mutations alter the order of the genes within the genome. Sturtevant~\cite{Sturtevant1917, Sturtevant1931} observed the biological phenomenon of genome rearrangement in the study of strains of \emph{Drosophila} (fruit flies), only a few years after he produced the first genetic map~\cite{Sturtevant1913}. Palmer~\&~Herbon~\cite{PalmerHerbon} observed similar phenomenon in plants. McClintock~\cite{mcclintock1951chromosome} also found experimental evidence of genes rearranging themselves, or ``transposing'' themselves, within chromosomes.  

Subsequent to his work on \textit{Drosophila}, Sturtevant together with Novitski~\cite{SturtevantNovitski} introduced one of the first genome rearrangement problems, seeking a minimum length sequence of operations, called a \emph{scenario}, 
that would transform one genome into another. In investigating these questions, it is of key importance to balance biological relevance with computational tractability. One central issue is that of combinatorial explosion: the number of scenarios even between small genomes may be too large to handle, making it difficult to test hypotheses on all possible scenarios. 
Thus, we desire a polynomial time algorithm to uniformly sample from the rearrangement scenarios. Since uniform sampling is no harder than exact enumeration~\cite{JERRUM1986169}, we investigate the computational complexity of the \algprobm{Pairwise Rearrangement} problem (Definition~\ref{def:DistancePairwiseRearrangement}) which asks for the number of minimum-length scenarios transforming one genome into another.  

The \algprobm{Pairwise Rearrangement} problem has received significant attention in several genome rearrangement models.
\algprobm{Pairwise Rearrangement} is known to be in $\textsf{FP}$ for the Single Cut or Join model~\cite{MiklosKissTannier}, but is conjectured to be \textsf{\#P}-complete for the Double Cut-and-Join model~\cite{MiklosTannierDCJFPRAS}. The Single Cut-and-Join model sits between these two models. Recently, Bailey, et al.~\cite{bailey2023complexity}, showed that \algprobm{Pairwise Rearrangement} is $\textsf{\#P}$-complete in the Single Cut-and-Join model. However, in practice, the key structures that serve as obstacles to efficient sampling may not necessarily appear. In particular, the relevant obstacles for efficient sampling in the Single Cut-and-Join model remain opaque. 

\noindent \\ \textbf{Main Results.} In this paper, we investigate the \algprobm{Pairwise Rearrangement} problem in the Single Cut-and-Join model through the lens of parameterized complexity. Our main result is the following.

\begin{theorem} \label{thm:MainFPT}
In the Single Cut-and-Join model, \algprobm{Pairwise Rearrangement} is fixed-parameter tractable with respect to the number of components in the adjacency graph (see Definition~\ref{def:adjgraph}) that are not \textit{trivial} (cycles of length $2$ or paths of length $1$).
\end{theorem}

Our parameterization in Theorem~\ref{thm:MainFPT} has biological significance. Indeed, chromoanagenesis is a carcinogenic mechanism that involves massive chromosomal rearrangements, which may lead to fewer nontrivial components in the adjacency graph~\cite{HollandCleveland}.

The \textit{adjacency graph} (see Definition~\ref{def:adjgraph}) is a bipartite multigraph illustrating where two genomes differ. Bergeron, Medvedev, and Stoye~\cite{BergeronMedvedevStoye} established a precise relationship between the adjacency graph and the distance between two genomes in the Single Cut-and-Join model. The operations in this model induce structural changes on the adjacency graph~\cite[Observation~2.7]{bailey2023complexity}. We leverage this crucially to establish Theorem~\ref{thm:MainFPT}. Our precise technique involves developing a dynamic programming algorithm that, when the number of nontrivial components is bounded, the corresponding lookup table only has a polynomial number of entries. This establishes our claim of polynomial-time computation. We stress that arriving at the recurrence relations for the dynamic programming algorithm and proving their correctness is technical and nontrivial. Indeed, this is not surprising, as \algprobm{Pairwise Rearrangement} is $\textsf{\#P}$-complete~\cite{bailey2023complexity}. While our work is theoretical in nature, the algorithm underlying Theorem~\ref{thm:MainFPT} in fact yields an efficient practical implementation (see \href{https://github.com/nhf216/scaj-fpt}{GitHub}).

We also note that if the distance (Equation~(\ref{eq:distance})) between the two genomes is bounded~\cite{BergeronMedvedevStoye}, then so is the number of nontrivial components in the adjacency graph. As a consequence, we obtain the following corollary:

\begin{corollary}\label{cor:distparam}
In the Single Cut-and-Join model, \algprobm{Pairwise Rearrangement} is fixed-parameter tractable with respect to the distance between two genomes.    
\end{corollary}

\begin{remark}
Using Theorem~\ref{thm:MainFPT}, we also obtain a fixed-parameter tractable algorithm to uniformly sample sorting scenarios. See Corollary~\ref{cor:Sampling}.
\end{remark}

To the best of our knowledge, parameterized complexity has received minimal attention in the genome rearrangement literature. For instance, a fixed-parameter tractable algorithm (parameterized by the number of components in the adjacency graph) for \algprobm{Pairwise Rearrangement} in the Double Cut-and-Join model can easily be deduced from the work of~\cite{BragaStoyeDCJ}, though the authors do not explicitly investigate the parameterized complexity of this problem.  Thus, beyond providing a means of coping with the intractability of \algprobm{Pairwise Rearrangement} in the Single Cut-and-Join model, Theorem~\ref{thm:MainFPT} (together with the results of~\cite{BragaStoyeDCJ}) makes precise that the number of nontrivial components in the adjacency graph serves as a key obstacle towards efficient sampling, across multiple models of genome rearrangement. 

In contrast, there has been significant algorithmic work on approximation and sampling (see, for instance,~\cite{MiklosKissTannier, MiklosSmithSamplingCounting, DarlingMiklosRagan, DurrettNielsenYork, MiklosTannier, LargetSimonKadaneSweet}), to cope with the intractability of enumeration. To the best of our knowledge, such approaches have not been fruitful against the Reversal model, for which the complexity of \algprobm{Pairwise Rearrangement} is a longstanding open problem. We are not aware of any work on approximate counting or sampling for \algprobm{Pairwise Rearrangement} in the Single Cut-and-Join model.

\section{Preliminaries} \label{sec:GenomeRearrangement}

We recall preliminaries regarding genome rearrangement. 

\begin{definition}
    A \emph{genome} is an edge-labeled directed graph in which each label is unique and the total degree of each vertex is 1 or 2 (in-degree and out-degree combined). In particular, a genome consists of disjoint paths and cycles. The weak components of a genome we call \emph{chromosomes}. Each edge begins at its \emph{tail} and ends at its \emph{head}, collectively referred to as its \emph{extremities}.  Degree $2$ vertices are called \emph{adjacencies}, and degree $1$ vertices are called \emph{telomeres}. {See Figure~\ref{adjacency-fig}.}
\end{definition}

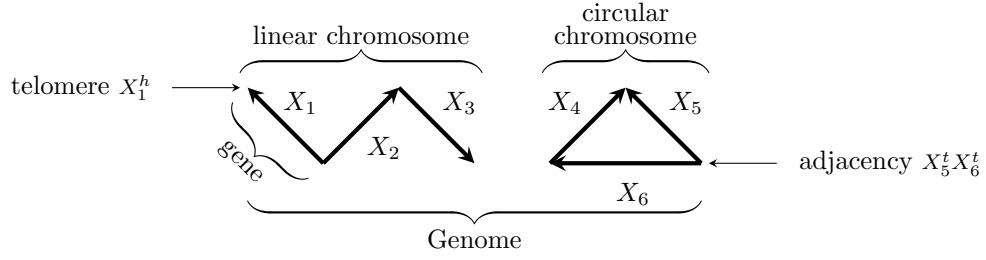
\begin{figure}[!ht]
\centering
\begin{tikzpicture}
\draw[ultra thick, >=stealth] (0,1) edge[<-] (1,0) (1,0) edge[->] (2,1) (2,1) edge[->] (3,0);
\draw[ultra thick, >=stealth] (4,0) edge[<-] (6,0) (6,0) edge [->] (5,1) (5,1) edge [<-] (4,0);
\draw (0.7,0.8) node {$X_1$} (1.8,0.2) node {$X_2$} (2.8,0.8) node {$X_3$} (4.2,0.8) node {$X_4$} (5.8,0.8) node {$X_5$} (5.1,-0.4) node {$X_6$};
\draw (-1,1) edge[->, >=stealth] (-0.1,1) (-2.2,1) node {telomere \footnotesize{$X_1^h$}} (6.1,0) edge[<-, >=stealth] (7,0) (8.5,0) node {adjacency \footnotesize{$X_5^tX_6^t$}};
\draw[thick, decorate, decoration={calligraphic brace, amplitude=3mm}] (-0.1,1.2) -- (3.1,1.2) (1.5,1.7) node {linear chromosome}; 
\draw[thick, decorate, decoration={calligraphic brace, amplitude=3mm}] (3.9,1.2) -- (6.1,1.2) (5,1.7) node {chromosome} (5,2) node {circular};
\draw[thick, decorate, decoration={calligraphic brace, mirror, amplitude=3mm}] (0,-0.5) -- (6,-0.5) (3,-1) node {Genome};
\draw[thick, decorate, decoration={calligraphic brace, mirror, amplitude=3mm}] (-0.2,0.8) -- (0.8,-0.2) (-0.1, 0) node[rotate=-45] {gene};
\end{tikzpicture}
\caption{An edge-labeled genome~\cite[Fig.~1]{bailey2023complexity}.} 
\label{adjacency-fig}
\end{figure}

Adjacencies can be viewed as sets of two extremities, and telomeres as sets containing exactly one extremity. For simplicity, we write adjacency $\{a,b\}$ as $ab$ or $ba$ and telomere $\{c\}$ as $c$.  For example, the adjacency $X_{5}^{t} X_{6}^{t}$ in Figure~\ref{adjacency-fig} denotes that the tail of the edge $X_{5}$ and the tail of the edge $X_{6}$ meet, and the telomere $X_1^h$ is where the edge $X_{1}$ ends. Each genome is then uniquely defined by its set of adjacencies and telomeres.

Consider the following operations on a given genome: 
\begin{enumerate}[itemsep=0pt]
    \item[ (i)] \emph{Cut}: an adjacency $ab$ is separated into two telomeres, $a$ and $b$,
    \item[ (ii)] \emph{Join}: two telomeres $a$ and $b$ become one adjacency, $ab$,
    \item[ (iii)] \emph{Cut-join}:   adjacency $ab$ and telomere $c$ are replaced with adjacency $ac$ and telomere $b$, and
    \item[ (iv)] \emph{Double-cut-join}:  adjacencies $ab$ and $cd$ are replaced with adjacencies $ac$ and $bd$.
\end{enumerate}

\begin{figure}[h!]
\centering
\begin{subfigure}{0.7\textwidth}
\centering
\begin{tikzpicture}
\draw[ultra thick, >=stealth, xshift=-0.5cm] (0,1) edge[<-] (1,0) (1,0) edge[->] (2,1) (2,1) edge[->] (3,0);
\draw[xshift=-0.5cm] (0.2,0.2) node {$X_1$} (1.8,0.2) node {$X_2$} (2.8,0.8) node {$X_3$};
\draw (3.5,0.5) edge [->, >=stealth] (5.5,0.5) (4.5,0.2) node {cut} (4.5,0.8) node {(i)};
\draw[ultra thick, >=stealth, xshift=0.5cm] (6,1) edge[<-] (7,0) (7,0) edge[->] (8,1) (9,1) edge[->] (9,0);
\draw[xshift=0.5cm] (6.2,0.2) node {$X_1$} (7.8,0.2) node {$X_2$} (8.6,0.6) node {$X_3$}; 
\end{tikzpicture}
\end{subfigure}

\vspace{.3in}
\begin{subfigure}{0.7\textwidth}
\centering
\begin{tikzpicture}
\draw[ultra thick, >=stealth, xshift=-0.5cm] (0,1) edge[<-] (1,0) (1,0) edge[->] (2,1) (2,1) edge[->] (3,0);
\draw[xshift=-0.5cm] (0.2,0.2) node {$X_1$} (1.8,0.2) node {$X_2$} (2.8,0.8) node {$X_3$};
\draw (3.5,0.5) edge [->, >=stealth] (5.5,0.5) (4.5,0.2) node {join} (4.5,0.8) node {(ii)};
\draw[ultra thick, >=stealth, xshift=0.5cm] (6,1) edge[<-] (7,0) (7,0) edge[->] (8,1) (8,1) edge[->] (6,1);
\draw[xshift=0.5cm] (6.2,0.2) node {$X_1$} (7.8,0.2) node {$X_2$} (7,.7) node {$X_3$};
\draw[white, xshift=0.5cm] (8,1) edge[->] (9,0);
\end{tikzpicture}
\end{subfigure}

\begin{subfigure}{0.7\textwidth}
\centering
\begin{tikzpicture}
\draw[ultra thick, >=stealth, xshift=-0.5cm] (0,1) edge[<-] (1,0) (1,0) edge[->] (2,1) (2,1) edge[->] (3,0);
\draw[xshift=-0.5cm] (0.2,0.2) node {$X_1$} (1.8,0.2) node {$X_2$} (2.8,0.8) node {$X_3$};
\draw (3.5,0.5) edge [->, >=stealth] (5.5,0.5) (4.5,0.2) node {cut-join} (4.5,0.8) node {(iii)};
\draw[ultra thick, >=stealth, xshift=0.5cm] (7,1) edge[<-, bend right=60, looseness=1.2] (7,0) (7,0) edge[->, bend right=60, looseness=1.2] (7,1) (8,1) edge[->] (9,0);
\draw[xshift=0.5cm] (6.2,0.2) node {$X_1$} (7.8,0.2) node {$X_2$} (8.8,0.8) node {$X_3$};
\draw[white, xshift=0.5cm] (7,1.5) node {$X_3$};
\end{tikzpicture}
\end{subfigure}

\begin{subfigure}{0.7\textwidth}
\centering
\begin{tikzpicture}
\draw[ultra thick, >=stealth, xshift=-0.5cm] (0,1) edge[<-] (1,0) (1,0) edge[->] (2,1) (2,1) edge[->] (3,0);
\draw[xshift=-0.5cm] (0.2,0.2) node {$X_1$} (1.8,0.2) node {$X_2$} (2.8,0.8) node {$X_3$};
\draw (3.5,0.5) edge [->, >=stealth] (5.5,0.5) (4.5,0.2) node {double-cut-join} (4.5,0.8) node {(iv)};
\draw[ultra thick, >=stealth, xshift=0.5cm] (6,1) edge[<-] (7,0) (8,1) edge[->] (7,0) (8,1) edge[->] (9,0);
\draw[xshift=0.5cm] (6.2,0.2) node {$X_1$} (7.8,0.2) node {$X_2$} (8.8,0.8) node {$X_3$}; 
\draw[white, xshift=0.5cm] (7,1.5) node {$X_3$};
\end{tikzpicture}
\end{subfigure}
\caption{(i) Adjacency $X_2^hX_3^t$ is cut. (ii) Telomeres $X_1^h$ and $X_3^h$ are joined. (iii) Adjacency $X_2^hX_3^t$ is cut, and resulting telomere $X_2^h$ is joined with $X_1^h$. (iv) Adjacencies $X_1^tX_2^t$ and $X_2^hX_3^t$ are
replaced with $X_1^tX_2^h$ and $X_2^tX_3^t$ \cite[Fig.~2]{bailey2023complexity}}. 
\label{operations-fig}
\end{figure}
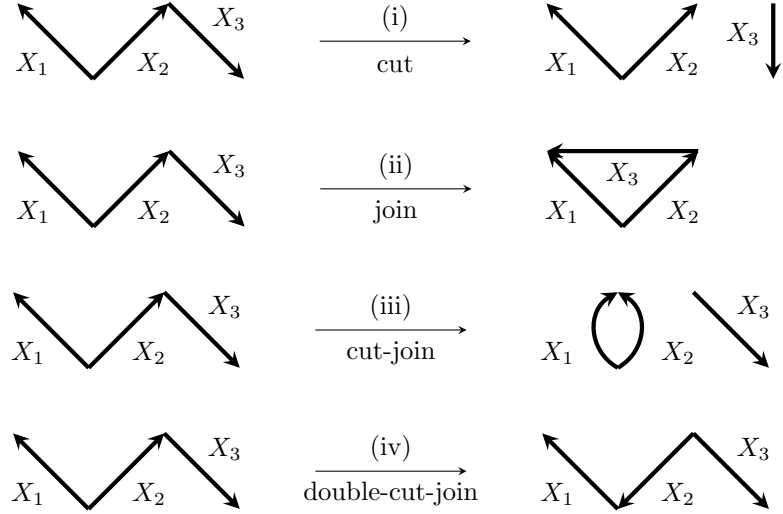

Note that a cut-join operation combines one cut and one join into a single operation, and a double-cut-join operation performs two cuts and two joins in one operation. See Figure~\ref{operations-fig} for an illustration of these operations. 

Several key models are based on these operations.  The \emph{Double Cut-and-Join~(DCJ)}~model was initially introduced by Yancopoulos,~Attie,~\&~Friedberg~\cite{YancopoulosAttieFriedberg} and permits all four operations. Later, \mbox{Feijão~\&~Meidanis~\cite{FeijaoMeidanis}} introduced the \emph{Single~Cut or Join~(SCoJ)}~model, which only allows operations (i) and (ii). Alternatively, the \emph{Single~Cut-and-Join~(SCaJ)}~model~\cite{BergeronMedvedevStoye} allows operations (i)--(iii), but not operation (iv). In this paper, we consider the Single Cut-and-Join model.

\begin{definition}\label{def:s}
    For any two genomes $G_1$ and $G_2$ with the same set of edge labels, there is a sequence of Single Cut-and-Join operations that transforms $G_1$ into $G_2$. Such a sequence is called a \emph{scenario}.
   The minimum possible length of such a scenario is called the \emph{distance} and is denoted $d(G_1, G_2)$. 
   An operation on a genome $G_1$ that (strictly) decreases the distance to genome $G_2$ is called a \emph{sorting operation} for $G_1$ and $G_2$. A scenario requiring $d(G_{1}, G_{2})$ operations to transform $G_{1}$ into $G_{2}$ is called a \emph{most parsimonious scenario} or \emph{sorting scenario}. When $G_{2}$ is understood, we refer to the action of transforming $G_{1}$ into $G_{2}$ using the minimum number of operations as \emph{sorting} $G_{1}$.  The number of most parsimonious scenarios transforming $G_{1}$ into $G_{2}$ is denoted $\mps(G_{1}, G_{2})$. 
\end{definition}

We now turn to defining the key algorithmic problem that we will consider in this paper.

\begin{definition} \label{def:DistancePairwiseRearrangement}
Let $G_{1}$ and $G_{2}$ be genomes. The \algprobm{Distance} problem asks to compute $d(G_{1}, G_{2})$. The \algprobm{Pairwise Rearrangement} problem asks to compute $\mps(G_1,G_2)$. 
\end{definition}

To investigate \algprobm{Pairwise Rearrangement}, we begin by introducing the adjacency graph. 

\begin{definition}\label{def:adjgraph}
    Given two genomes $G_1$ and $G_2$ with the same set of edge labels, the \emph{adjacency graph} $A(G_1,G_2)$ is a bipartite undirected multigraph $(V_{1} \dot\cup V_{2}, E)$ 
    where the vertices in $V_{i}$ are the adjacencies and telomeres in $G_i$ and for any $X\in V_1$ and $Y\in V_2$, the number of edges between $X$ and $Y$ is $|X\cap Y|$. 
\end{definition}

\begin{figure}[!h]
\centering
\begin{minipage}{0.8\textwidth}
\centering
\hspace{2.1cm}\begin{tikzpicture}[scale=1.2, every node/.style={scale=1.2}]
\draw[ultra thick, >=stealth] (0,1) edge[<-] (1,0) (1,0) edge[->] (2,1) (2,1) edge[<-] (3,0) (3,0) edge[<-] (4,1);
\draw[ultra thick, >=stealth] (5,0) edge[<-] (7,0) (7,0) edge [->] (6,1) (6,1) edge [<-] (5,0);
\draw (0.2,0.2) node {$X_1$} (1.8,0.2) node {$X_2$} (2.8,0.8) node {$X_3$} (3.8,0.2) node {$X_4$} (5.2,0.8) node {$X_5$} (6.8,0.8) node {$X_6$} (6.1,-0.4) node {$X_7$};
\draw[thick] (-1,0.5) node {$G_1$};
\end{tikzpicture}
\end{minipage}

\begin{minipage}{0.8\textwidth}
\centering
\begin{tikzpicture}[scale=1.2, every node/.style={scale=1.2}]
\draw[thick] (0,0) -- (0,-2) (3,0)-- (1.9,-2) -- (1,0) -- (0.9,-2) -- (2,0) --(2.9,-2) --(3,0)  (4,0) -- (3.75,-2) (5.25,-2) edge[bend right=40] (5,0) (5,0) edge[bend right=40] (5.25,-2) (4.4,-2) -- (6,0) -- (6.95,-2) (7,0) edge[bend left=40] (6.1,-2) (7,0) edge[bend right=40] (6.1,-2); 
\draw[thick] (0,0) node [myStyle] {} (1,0) node [myStyle] {} (2,0) node [myStyle] {} (3,0) node [myStyle] {} (4,0) node [myStyle] {} (5,0) node [myStyle] {} (6,0) node [myStyle] {} (7,0) node [myStyle] {};
\draw[thick] (0,-2) node [myStyle] {} (0.9,-2) node [myStyle] {} (1.9,-2) node [myStyle] {} (2.9,-2) node [myStyle] {} (3.75,-2) node [myStyle] {} (4.4,-2) node [myStyle] {} (5.25,-2) node [myStyle] {} (6.1,-2) node [myStyle] {} (6.95,-2) node [myStyle] {};
\draw[thick, white] (-1.5,0.5) node {$G_1$};
\draw[thick] (0,0.3) node {\footnotesize{$X_1^h$}} (1,0.3) node {\footnotesize{$X_1^tX_2^t$}} (2,0.3) node {\footnotesize{$X_2^hX_3^h$}} (3,0.3) node {\footnotesize{$X_3^tX_4^h$}} (4,0.3) node {\footnotesize{$X_4^t$}} (5,0.3) node {\footnotesize{$X_5^tX_7^h$}} (6,0.3) node {\footnotesize{$X_5^hX_6^h$}} (7,0.3) node {\footnotesize{$X_6^tX_7^t$}};
\draw[thick] (0,-2.3) node {\footnotesize{$X_1^h$}} (0.9,-2.3) node {\footnotesize{$X_1^tX_2^h$}} (1.9,-2.3) node {\footnotesize{$X_2^tX_3^t$}} (2.9,-2.3) node {\footnotesize{$X_3^hX_4^h$}} (3.75,-2.3) node {\footnotesize{$X_4^t$}} (4.4,-2.3) node {\footnotesize{$X_5^h$}} (5.25,-2.3) node {\footnotesize{$X_5^tX_7^h$}} (6.1,-2.3) node {\footnotesize{$X_6^tX_7^t$}} (6.95,-2.3) node {\footnotesize{$X_6^h$}};
\draw (-1.0,0.3) node {$V_1$};
\draw  (-1.0,-2.3) node {$V_2$};
\draw[thick, decorate, decoration={calligraphic brace, amplitude=3mm}] (-1.4,-2.5) -- (-1.4,0.5) (-2.6,-1.0) node {$A(G_1,G_2)$}; 
\end{tikzpicture}
\end{minipage}

\begin{minipage}{0.8\textwidth}
\centering
\hspace{2.1cm}\begin{tikzpicture}[scale=1.2, every node/.style={scale=1.2}]
\draw[ultra thick, >=stealth] (0,1) edge[<-] (1,0) (2,1) edge[->] (1,0) (2,1) edge[->] (3,0) (3,0) edge[<-] (4,1);
\draw[ultra thick, >=stealth] (5.3,0) edge[<-] (6.5,0) (6.5,0) edge [->] (7,1) (4.8,1) edge [<-] (5.3,0);
\draw (0.2,0.2) node {$X_1$} (1.8,0.2) node {$X_2$} (2.8,0.8) node {$X_3$} (3.8,0.2) node {$X_4$} (5.4,0.5) node {$X_5$} (6.4,0.5) node {$X_6$} (6.1,-0.4) node {$X_7$};
\draw[thick] (-1,0.5) node {$G_2$};
\draw[thick, white] (0.5,1.4) node {$X_6$};
\end{tikzpicture}
\end{minipage}
\caption{ The adjacency graph $A(G_1,G_2)$ is shown in the middle, for genomes $G_1$ and $G_2$ shown above and below, respectively \cite[Fig.~3]{bailey2023complexity}.}
\label{fig:adj graph}
\end{figure}

Note that each vertex in an adjacency graph $A(G_1,G_2)$ must have either degree 1 or 2 (corresponding, respectively, to telomeres and adjacencies in the original genome), and so  $A(G_1,G_2)$ is composed entirely of disjoint cycles and paths. Note also that every operation on $G_1$ corresponds to an operation on $V_1$ in $A(G_1,G_2)$.  For example, in Figure~\ref{fig:adj graph} the cut operation on $G_1$ which separates adjacency $X_5^hX_6^h$ into telomeres $X_5^h$ and $X_6^h$ equates to separating the corresponding vertex $X_5^hX_6^h$ in $V_1$ into two vertices $X_5^h$ and $X_6^h$, thus splitting the path of length 2 in $A(G_1,G_2)$ into two disjoint paths of length 1. In a similar fashion, a join operation on $G_1$ corresponds to combining two vertices $a$, $b$ in $V_1$ into a single vertex $ab$, and a cut-join operation on $G_1$ corresponds to replacing vertices $ab$, $c$ in $V_1$ with vertices $ac$, $b$. Whether or not an operation on $A(G_{1}, G_{2})$ corresponds to a sorting operation on $G_1$---that is, whether it decreases the distance to $G_2$ or not---depends highly on the structure of the components upon which we are acting. 
To better describe such sorting operations, we adopt the following classification of components in $A(G_1,G_2)$ and notion of their size:

\begin{definition}\label{defn:size} The possible connected components of $A(G_{1}, G_{2})$ are classified as follows:
\begin{itemize}
    \item A $W$-\emph{shaped component} is an even path with its two endpoints in $V_1$. 
    \item An $M$-\emph{shaped component} is an even path with its two endpoints in $V_2$.
    \item An $N$-\emph{shaped component} is an odd path.
    \item A \emph{crown} is an even cycle.
\end{itemize}  \end{definition}

\begin{definition}
    The \emph{size} of a component $B$ in $A(G_1,G_2)$ is defined to be $\lfloor \, |E(B)|/2\, \rfloor$. We refer to an $N$-shaped component of size 0 (a single edge) as a \emph{trivial path}, and a crown of size 1 (a 2-cycle) as a \emph{trivial crown}.
\end{definition}

The language ``trivial'' is motivated by the fact that such components indicate where $G_1$ and $G_2$ already agree, and hence no sorting operations are required on vertices belonging to trivial components (see, e.g., the trivial components in Figure~\ref{fig:adj graph}). Indeed, a sorting scenario can be viewed as a minimal length sequence of operations which produces an adjacency graph consisting of only trivial components. Notably, a scenario is a sorting scenario if and only if it consists exclusively of sorting operations. We now list precisely the sorting operations below.

\begin{observation}[{\cite[Observation~2.7]{bailey2023complexity}}] \label{obs:SortingScenarios}
\noindent In the SCaJ model, a case analysis of all operations yields precisely these as the only   
sorting operations on $A(G_1,G_2)$:
\begin{enumerate}[label=(\alph*)]
\item A cut-join operation on a nontrivial $N$-shaped component, producing an $N$-shaped component and a trivial crown
\item A cut-join operation on a $W$-shaped component of size at least 2, producing a trivial crown and a $W$-shaped component
\item A join operation on a $W$-shaped component of size 1, producing a trivial crown
\item A cut operation on an $M$-shaped component, producing two $N$-shaped components
\item A cut operation on a nontrivial crown, producing a $W$-shaped component
\item A cut-join operation on an $M$-shaped component and a $W$-shaped component, where an adjacency in the $M$-shaped component is cut and joined to a telomere in the $W$-shaped component, producing two $N$-shaped components
 \item A cut-join operation on a nontrivial crown and an $N$-shaped component, where an adjacency in the crown is cut and joined to the telomere in the $N$-shaped component, producing an $N$-shaped component
\item A cut-join operation on a nontrivial crown and a $W$-shaped component, where an adjacency from the crown is cut and joined with a telomere from the $W$-shaped component, producing a $W$-shaped component
\end{enumerate}
\end{observation}

\begin{figure}[!ht]
\centering
\begin{tikzpicture}[scale=0.9]
\draw[ultra thick, >=stealth, xshift=-0.5cm] (0,1) edge[-] (0,.5) (0,.5) edge[->] (-.3,0) (0,.5) edge[->] (.3,0);
\draw[xshift=-0.5cm] (0,1.2) node {$N$} (-.5,-0.2) node {$N$} (.5,-0.2) node {$T$} (-1.3,1.7) node {(I)};
\draw[ultra thick, >=stealth, xshift=2.0cm] (0,1) edge[-] (0,.5) (0,.5) edge[->] (-.3,0) (0,.5) edge[->] (.3,0);
\draw[xshift=2.0cm] (0,1.5) node {$W$} 
(0,1.2) node {(size $>1$)}
(-.5,-0.2) node {$W$} (.5,-0.2) node {$T$}; 
\draw [ultra thick, >-stealth,xshift=4.5cm, double distance=3pt,
             arrows = {-Latex[length=0pt 2 .3]}] (0,1) -- (0,0);
\draw[xshift=4.5cm] (0,1.5) node {$W$} (0,1.2) node {(size 1)} (0,-0.2) node {$T$}; 
\draw[ultra thick, >=stealth, xshift=7.0cm] (0,1) edge[dashed,-] (0,.5) (0,.5) edge[dashed,->] (-.3,0) (0,.5) edge[dashed,->] (.3,0);
\draw[xshift=7.0cm] (0,1.2) node {$M$} (-.5,-0.2) node {$N$} (.5,-0.2) node {$N$}; 
\draw[ultra thick, >=stealth, xshift=-0.5cm,yshift=-2.5cm] (0,1) edge[dashed,->] (0,0) ;
\draw[xshift=-0.5cm,yshift=-2.5cm] (0,1.2) node {$C$} (0,-0.2) node {$W$} ; 
 \draw[ultra thick, >=stealth, xshift=2.0cm,yshift=-2.5cm] (0,1) edge[-] (0,.5) (0,.5) edge[->] (-.3,0) (0,.5) edge[->] (.3,0);
 \draw[xshift=2.0cm,yshift=-2.5cm] (0,1.2) node {$M \sim W$} (-.5,-0.2) node {$N$} (.5,-0.2) node {$N$}; 
 \draw[ultra thick, >=stealth, xshift=4.5cm,yshift=-2.5cm] (0,1) edge[->] (0,0) ;
 \draw[xshift=4.5cm,yshift=-2.5cm] (0,1.2) node {$C \sim N$} (0,-0.2) node {$N$} ; 
\draw[ultra thick, >=stealth, xshift=7.0cm,yshift=-2.5cm] (0,1) edge[->] (0,0) ;
\draw[xshift=7.0cm,yshift=-2.5cm] (0,1.2) node {$C \sim W$} (0,-0.2) node {$W$} ; 
\draw[xshift=9.5cm,yshift=-1cm] (0,1) edge[->] (1,1) 
(1.5,1) edge[->] (2.5,.35) 
(0,-0.5) edge[->] (1,-0.5) 
(1.5,-0.5) edge[->] (2.5,.15) 
(1.25,.8) edge[->] (1.25,-.25) 
(0,.8) edge[->] (1,-.3) 
;
\draw[xshift=9.5cm,yshift=-1cm] (-1.2,3.0) rectangle (3.2,-2.0); 
\draw[xshift=9.5cm,yshift=-1cm] (-11.7,3.0) rectangle (3.2,-2.0); 
\draw[->, xshift=9.5cm, yshift=-1cm] (1.0,1.2) arc (240:-55:.4); 
\draw[<-, xshift=9.5cm, yshift=-1cm] (1.4,-0.7) arc (65:-240:.4); 
\draw[xshift=9.5cm,yshift=-1cm] (-.3,1) node {$C$} (1.25,1) node {$W$} (2.75,.25) node {$T$} (-.3,-0.5) node {$M$} (1.25,-0.5) node {$N$}
(-.8,2.7) node {(II)};
\end{tikzpicture}
\caption{This figure depicts the operations described in Observation~\ref{obs:SortingScenarios}, where the arrows point from the original component type(s) to the component type(s) produced by the three operations allowed in the SCaJ model: cut, join, and cut-join. Here, $T$ denotes a trivial crown and $C$ denotes a nontrivial crown. The eight diagrams in (I) show each of the sorting operations (a)--(h) where bold single arrows represent cut-join operations, dashed arrows represent cut operations, and the double arrow represents the join operation. Diagram (II) summarizes which components can be produced.  Note that all operations will only result in $W$-shaped components, $N$-shaped components, and/or trivial crowns. }
\label{operations-drawing}
\label{fig:sortingops}
\end{figure}
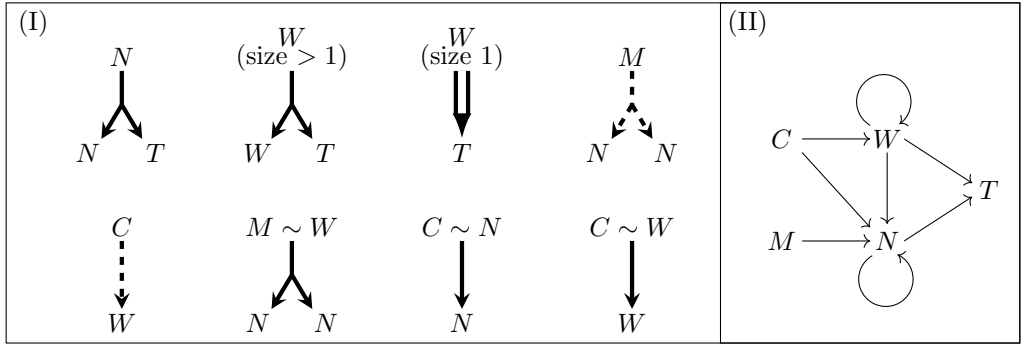

Note that (a)--(e) are sorting operations on $G_1$ that operate on only one component in the adjacency graph, though they may produce two different components. On the other hand, (f)--(h) are sorting operations on $G_1$ that operate on two separate components in the adjacency graph. See Figure~\ref{fig:sortingops} for a visualization of each sorting operation.

Using these sorting operations on $A(G_1,G_2)$, the distance between two genomes $G_1$ and $G_2$ for the SCaJ model is given by
\begin{equation} \label{eq:distance}
    d(G_1,G_2) = n - \frac{\#N}{2} - \#T + \#C
\end{equation}
where $n$ is the number of edges in $G_1$ (equivalently, one half of the number of edges in $A(G_1,G_2)$), $\#N$ is the number of $N$-shaped components, $\#T$ is the number of trivial crowns, and $\#C$ is the number of nontrivial crowns~\cite{BergeronMedvedevStoye}.  

Let $\mathcal{B}$ be the set of all components of $A(G_1,G_2)$ and let $\mathcal{B}'$ be a subset of $\mathcal{B}$. Define
\begin{align} 
d(\mathcal{B}'):=\left(\sum_{B\in\mathcal{B}'}\text{size}(B)\right)-\#T_{\mathcal{B}'}+\#C_{\mathcal{B}'} \label{eq:SubsetDistance}
\end{align} 
where $\#T_{\mathcal{B}'}$ and $\#C_{\mathcal{B}'}$ are the number of trivial crowns and nontrivial crowns in $\mathcal{B}'$, respectively. The quantity $d(\mathcal{B}')$ is the minimum number of operations needed to transform all components of $\mathcal{B}'$ into trivial components, with no operation acting on a component not belonging to $\mathcal{B}'$. 
Note that $d(\mathcal{B})=d(G_1,G_2)$, as the $\frac{\#N}2$ term is absorbed into the summation of the sizes of all components.

\begin{definition}[{\cite[Definition~2.8]{bailey2023complexity}}]\label{def:equivrelation}
Let $A$ and $B$ be components of an adjacency graph, and consider a particular sorting scenario. We say $A\sim B$ if either $A=B$ or there is a cut-join operation in the scenario where an extremity $a$ from $A$ and an extremity $b$ from $B$ are joined into an adjacency. The transitive closure of $\sim$ is an equivalence relation which we call \emph{sort together}. We will be particularly interested in subsets of the equivalence classes of ``sort together.'' We abuse terminology by referring to such a subset as a set that \emph{sorts together}.
\end{definition}

Note that if two components $A$ and $B$ in $A(G_1, G_2)$ satisfy $A\sim B$, the cut-join operation does not need to occur immediately. For example, two nontrivial crowns $C_1$ and $C_2$ can satisfy $C_1\sim C_2$ by first cutting $C_1$ to produce a $W$-shaped component, then operation (b) can be applied one or more times before operation (h) sorts $C_2$ and the remaining $W$-shaped component together; see Figure~\ref{operations-drawing}. 

We will now recall additional notation from~\cite{bailey2023complexity} that we will use in this paper. For a subset $\mathcal{B'}$ of $\mathcal{B}$, define $\mps(\mathcal{B'})$ as the number of sequences with $d(\mathcal{B'})$ operations in which the components of  $\mathcal{B}'$ are transformed into trivial components with no operation acting on a component not belonging to $\mathcal{B}'$.  Note that $\mps(\mathcal{B})$ is the number of most parsimonious scenarios transforming $G_{1}$ into $G_{2}$.  
Let $\stg( \mathcal{B}')$ denote the number of sequences with $d(\mathcal{B}')$ operations in which the components of  $\mathcal{B}'$ sort together and are transformed into trivial components with no operation acting on a component not belonging to $\mathcal{B}'$.

Note that if $\mathcal{B}'$ and $ \mathcal{B}''$ are two subsets of $\mathcal{B}$ that have all the same component types with all the same sizes, then $\mps(\mathcal{B}')=\mps(\mathcal{B}'')$ and $\stg(\mathcal{B}')=\stg(\mathcal{B}'')$. Going forward, we will often care about values of $\mps$ and $\stg$ only in the context of their component types and sizes. Suppose we are given multisets $\mathcal{C}$, $\mathcal{M}$, $\mathcal{W}$, and $\mathcal{N}$ of nonnegative integers with every element of $\mathcal{C}$ at least $2$ and every element of $\mathcal{M}$ and $\mathcal{W}$ at least $1$. We define $\mps(\mathcal{C},\mathcal{M},\mathcal{W},\mathcal{N})$ to equal $\mps(\mathcal{B}')$ for any set of components $\mathcal{B}'$ with $\ab{\mathcal{C}}$ nontrivial crowns of sizes in $\mathcal{C}$, $\ab{\mathcal{M}}$ $M$-shaped components of sizes in $\mathcal{M}$, $\ab{\mathcal{W}}$ $W$-shaped components of sizes in $\mathcal{W}$, and $\ab{\mathcal{N}}$ $N$-shaped components of sizes in $\mathcal{N}$. We define $\stg(\mathcal{C},\mathcal{M},\mathcal{W},\mathcal{N})$ similarly.

\section{Combinatorics of Genome Rearrangement}

In this section, we will develop key combinatorial tools that will be used in Section~\ref{sec:FPT} to prove Theorem~\ref{thm:MainFPT}. Our strategy will be to build a lookup table for dynamic programming. In particular, our technique relies crucially on the following lemma.

\begin{lemma}[{\cite[Lemma~A.4]{bailey2023complexity}}] \label{lem:partcount}
Let $\mathcal{B'}$ be a subset of components of an adjacency graph, and let $\Pi(\mathcal{B'})$ denote the set of all partitions of $\mathcal{B'}$. We have
\[
\mps(\mathcal{B'})=\sum_{\pi\in\Pi(\mathcal{B'})}\binom{d(\mathcal{B'})}{d(\pi_1),d(\pi_2),\ldots,d(\pi_{p(\pi)})}\prod_{i=1}^{p(\pi)}\stg(\pi_i),
\]
where $\pi = \{\pi_1,\pi_2, \ldots, \pi_{p(\pi)}\}$.
\end{lemma}

We will utilize Lemma~\ref{lem:partcount} in the following manner. Fix an entry in the lookup table, and let $\mathcal{B}'$ denote the set of components being considered at said entry. In order to compute $\mps(\mathcal{B}')$, we will proceed as follows. Fix a partition $\pi \in \Pi(\mathcal{B}')$. For each $i \in [p(\pi)]$, we first check if $\stg(\pi_{i}) = 0$. This step is computable in polynomial-time by checking whether there exists a permissible sorting operation (see Theorem~\ref{thm:invalidsort}). If $\stg(\pi_{i}) \neq 0$, then we access the entry in the lookup table for $\stg(\pi_{i})$. This will allow us to compute 
\[
\binom{d(\mathcal{B'})}{d(\pi_1),d(\pi_2),\ldots,d(\pi_{p(\pi)})}\prod_{i=1}^{p(\pi)}\stg(\pi_i).
\]
We will show later that as the number of nontrivial crowns in the adjacency graph is bounded (by assumption), there are only a polynomial number of partitions $\pi = (\pi_1, \ldots, \pi_{p(\pi)}) \in \Pi(\mathcal{B}')$ such that $\stg(\pi_{i}) \neq 0$ for all $i \in [p(\pi)]$.

We will now investigate how to compute $\stg(\pi_{i})$, which requires studying which sets of components can and cannot sort together.

\begin{proposition}\label{prop:invalidsort}
Given a set $\compset$ of components of some adjacency graph $G$, $\stg\p{\compset}=0$ if $\compset$ has any of the following properties:
\begin{enumerate}
\item $\compset$ contains at least two components, at least one of which is a trivial crown. 
\item $\compset$ contains more than one $W$-shaped component.
\item $\compset$ contains more than one $M$-shaped component.
\item $\compset$ contains more than one path and at least one $N$-shaped component.
\end{enumerate}
\end{proposition}

\begin{proof}
We refer to Observation~\ref{obs:SortingScenarios} and Figure~\ref{fig:sortingops} for the permissible sorting operations, from which the proof essentially follows. We provide full details below.
\begin{enumerate}
\item None of the sorting operations (a)--(h) sort a trivial crown with any other component. So if $\mathcal{B'}$ contains at least one trivial crown and any other component, then $\stg(\compset)=0$, as required.

\item  Suppose $\compset$ contains components $A$ and $B$, where $A$ is a $W$-shaped component.  Note that $A\sim B$ is only possible if $B$ is an $M$-shaped component (sorting operation (f)) or a nontrivial crown (sorting operation (h)). As seen in Figure~\ref{fig:sortingops} Diagram (II), sorting operations with $W$-shaped components cannot create an $M$-shaped component or a nontrivial crown. So, if $B$ is also a $W$-shaped component, any sequence of sorting operations applied to $A$ and $B$ will not yield components that can sort together. Therefore $\stg(\compset)=0$.

\item Suppose $\compset$ contains components $A$ and $B$, where $A$ is an $M$-shaped component. Note that $A\sim B$ is only possible if $B$ is a $W$-shaped component (sorting operation (f)) or a nontrivial crown  (sorting operation (d) followed by (g)). As seen in Figure~\ref{fig:sortingops} Diagram (II), sorting operations with $M$-shaped components cannot create a $W$-shaped component or a nontrivial crown. Therefore, if $B$ is also an $M$-shaped component, any sequence of sorting operations applied to $A$ and $B$ will not yield components that can sort together. Therefore $\stg(\compset)=0$.  Note that this can also be obtained by Corollary~\ref{cor:mwsymmetry} below.

\item Suppose $\compset$ contains components $A$ and $B$, where $A$ is an $N$-shaped component. Note that $A\sim B$ is only possible if $B$ is a nontrivial crown, using sorting operation (g). As seen in Figure~\ref{fig:sortingops} Diagram (II), sorting operations with any component that is a path cannot create a nontrivial crown. Therefore, if $B$ is a path, any sequence of sorting operations applied to $A$ and $B$ will not yield components that can sort together. Therefore $\stg(\compset)=0$.  \qedhere
\end{enumerate}
\end{proof}

Indeed, we can further strengthen Proposition~\ref{prop:invalidsort} to an equivalence statement as follows: 

\begin{theorem} \label{thm:invalidsort}
Given a set $\compset$ of components of some adjacency graph $G$, $\stg\p{\compset}=0$ if and only if $\compset$ has any of the following properties:
\begin{enumerate}
\item $\compset$ contains at least two components, at least one of which is a trivial crown. 
\item $\compset$ contains more than one $W$-shaped component.
\item $\compset$ contains more than one $M$-shaped component.
\item $\compset$ contains more than one path and at least one $N$-shaped component.
\end{enumerate}
\end{theorem}

\begin{proof}
Proposition~\ref{prop:invalidsort} established one direction of this proof. For the other direction, we will use a proof by contrapositive. 
Suppose that $\mathcal{B'}$ does not have any of the four properties listed in the theorem. Therefore, $\mathcal{B'}$ must be one of the following types: 
\begin{enumerate}
    \item [(i)] one $N$-shaped component and zero or more nontrivial crowns,
    \item [(ii)] one $W$-shaped component and zero or more nontrivial crowns,
    \item [(iii)] one or more nontrivial crowns,
    \item[(iv)] one $M$-shaped component and zero or more nontrivial crowns, 
    \item[(v)] one $M$-shaped component, one $W$-shaped component, and zero or more nontrivial crowns, or
    \item[(vi)] one trivial crown.
\end{enumerate}
We now use the operations in Observation~\ref{obs:SortingScenarios} to exhibit a sorting scenario for each type of $\mathcal{B'}$ which sorts the components together. This shows $\stg(\mathcal{B'}) >0$. 

\begin{enumerate}
    \item[(i)] As long as $\mathcal{B}'$ contains at least one nontrivial crown, then we use a cut-join operation, as in Observation~\ref{obs:SortingScenarios}(g) on the $N$-shaped component and one of the nontrivial crowns to create an $N$-shaped component. This step is repeated until there are no more nontrivial crowns. The single $N$-shaped component can be sorted using cut-join operations Observation~\ref{obs:SortingScenarios}(a).

    \item[(ii)] As long as $\mathcal{B}'$ contains at least one nontrivial crown, then we use a cut-join operation, as in Observation~\ref{obs:SortingScenarios}(h), to transform the $W$-shaped component and one of the nontrivial crowns into a single $W$-shaped component. This step can be repeated until there are no more nontrivial crowns. Then, the single $W$-shaped component can be sorted with a series of cut-join operations from Observation~\ref{obs:SortingScenarios}(b) followed by one join operation from Observation~\ref{obs:SortingScenarios}(c) when the $W$-shaped component is size $1$.

    \item[(iii)] We first cut one of the nontrivial crowns to produce a $W$-shaped component, as in Observation~\ref{obs:SortingScenarios}(e), and then we proceed as in (ii) above. 
    
    \item[(iv)] We first cut the $M$-shaped component to produce two $N$-shaped components, as in Observation~\ref{obs:SortingScenarios}(d), and then we complete the scenario as in (i) above. 

    \item[(v)] We  use Observation~\ref{obs:SortingScenarios}(f) to cut-join the $M$-shaped component and $W$-shaped component to produce two $N$-shaped components. Then we use a series of cut-join operations, as in Observation~\ref{obs:SortingScenarios}(a), to sort the smaller $N$-shaped component into trivial components. And finally, we follow (i) to sort the remaining $N$-shaped component with the (zero or more) nontrivial crowns. 

    \item[(vi)] $\mathcal{B'}$ is already sorted, so no operations are needed and $\stg(\mathcal{B'})=1$. \qedhere
\end{enumerate}
\end{proof}

Aside from the case where $\mathcal{B'}$ consists of just one trivial crown (case (vi) from the proof of Theorem~\ref{thm:invalidsort}), we list each of the other cases when $\stg(\mathcal{B}') \neq 0$ and define a function along with each  to simplify $\stg(\mathcal{C},M,W,N)$. Note that in what follows, $\mathcal{Z}_{\geq2}$ denotes the set of finite multisets of integers in which each integer is at least $2$, and in this paper we take $\mathbb{N}$ to include $0$.

\begin{itemize}
\item[(i)] The components are a single $N$-shaped component and zero or more nontrivial crowns. Define  $\stg_N:\mathcal{Z}_{\geq2}\times\mathbb{N}\to\mathbb{N}$ as $\stg_N(\mathcal{C},\eta)=\stg(\mathcal{C},\emptyset,\emptyset,\{\eta\})$. 
\item[(ii)] The components are a single $W$-shaped component and zero or more nontrivial crowns. Define $\stg_W:\mathcal{Z}_{\geq2}\times\mathbb{Z}^+\to\mathbb{N}$ as $\stg_W(\mathcal{C},w)=\stg(\mathcal{C},\emptyset,\{w\},\emptyset)$.
\item[(iii)] The components are one or more nontrivial crowns. Define $\stg_{C}:\mathcal{Z}_{\geq2} - \{ \emptyset\} \to\mathbb{N}$ as $\stg_C(\mathcal{C})=\stg(\mathcal{C},\emptyset,\emptyset,\emptyset)$.
\item[(iv)] The components are a single $M$-shaped component and zero or more nontrivial crowns. Define $\stg_M:\mathcal{Z}_{\geq2}\times\mathbb{Z}^+\to\mathbb{N}$ as $\stg_M(\mathcal{C},m)=\stg(\mathcal{C},\{m\},\emptyset,\emptyset)$.
\item[(v)] The components are a single $M$-shaped component, a single $W$-shaped component, and zero or more nontrivial crowns. 
 Define $\stg_{MW}:\mathcal{Z}_{\geq2}\times\mathbb{Z}^+\times\mathbb{Z}^+\to\mathbb{N}$ as $\stg_{MW}(\mathcal{C},m,w)=\stg(\mathcal{C},\{m\},\{w\},\emptyset)$.
To  allow us to gracefully take advantage of a symmetry between $M$-shaped and $W$-shaped components, we extend the domain of $\stg_{MW}$ to $\mathcal{Z}_{\geq2}\times\mathbb{N}\times\mathbb{N}$ by introducing the convention that $\stg_{MW}(\mathcal{C},m,w)=0$ when $m=0$ or $w=0$.
\end{itemize}

We recall the following lemma, which allows us to compute the number of sorting scenarios for a single component.

\begin{lemma}[{\cite[Lemma~2.9]{bailey2023complexity}}] \label{lem:SortOneComponent}$\,$
\begin{itemize}
    \item For all $\eta\in\mathbb{N}$, $\stg_N(\emptyset,\eta)=1$.
    \item For all $w\in\mathbb{Z}^+$, $\stg_W(\emptyset,w)=2^{w-1}$.
    \item For all $m\in\mathbb{Z}^+$, $\stg_M(\emptyset,m)=2^{m-1}$.
    \item For all $c\in\mathbb{Z}_{\geq2}$, $\stg_C(\st{c})= c\cdot 2^{c-1}$.
\end{itemize}
\end{lemma}

Our goal now is to enumerate the number of sorting scenarios in each of these cases. We first provide a recurrence relation for $\stg_{N}(\mathcal{C},\eta)$.

\begin{proposition}\label{prop:recn}
We have the following recurrence relations for $\stg_N$:
\[ \stg_{N}(\mathcal{C},\eta) = \begin{cases} 1 & : \mathcal{C} = \emptyset,  \eta = 0, \\\sum_{c \in \mathcal{C}} \left(2c \cdot \stg_{N}(\mathcal{C} - \{c\}, c)\right) & : \mathcal{C} \neq \emptyset, \eta = 0, \\\stg_{N}(\mathcal{C},\eta-1) + \sum_{c \in \mathcal{C}} \left(2c \cdot \stg_{N}(\mathcal{C} - \{c\}, \eta+c)\right)  & : \text{otherwise}.\end{cases}\]
\end{proposition}

\begin{proof}
First, $\stg_N(\emptyset,0)$ is the case of a single path of size $0$ (a single edge). This represents that the given  telomere has already been sorted, so there is only one way to sort this component (do nothing). On the other hand, when sorting a collection of nontrivial crowns together with an $N$-shaped component, the first operation either consists of applying a cut-join on the $N$-shaped component alone if it has size greater than $0$ ($1$ way to do this), or applying a cut-join of a nontrivial crown to the $N$-shaped component. If the crown being operated on has size $c$, there are $c$ possible places to cut it, and there are $2$ possible ways to join it to to the $N$-shaped component (either telomere of the newly cut crown). The result of this operation is one fewer crown and an increase in the size of the $N$-shaped component by $c$. These considerations lead to the recursive formulas for $\stg_N(\mathcal{C},\eta)$.
\end{proof}

We next provide a recurrence relation for $\stg_{W}(\mathcal{C},w)$.

\begin{proposition}\label{prop:recw}
We have the following recurrence relations for $\stg_W$:
\[
\stg_{W}(\mathcal{C},w) = \begin{cases}
    1 & \hspace{-0.05in}: \mathcal{C} = \emptyset, w = 1, \\
    \sum_{c \in \mathcal{C}} \left(4c \cdot \stg_{W}(\mathcal{C} - \{c\}, c + 1) \right) & \hspace{-0.05in} : \mathcal{C} \neq \emptyset, w = 1, \\
    2 \cdot \stg_{W}(\mathcal{C},w-1) + \sum_{c \in \mathcal{C}} \left(4c \cdot \stg_{W}(\mathcal{C} - \{c\}, w+c)\right)  & : \text{otherwise}.
\end{cases}
\]
\end{proposition}

\begin{proof}
First, $\stg_W(\emptyset,1)$ is the case of a single $W$-shaped component of size $1$. To sort such a path, join the telomeres in the top genome. So, there is only one way to sort this component. On the other hand, when sorting a collection of nontrivial crowns together with a $W$-shaped component, the first operation either consists of a cut-join of one end of the $W$-shaped component if it has size greater than $1$ ($2$ ways to do this) or a cut-join of a nontrivial crown to the $W$-shaped component. If the crown being operated on has size $c$, there are $c$ possible places to cut it, and there are $4$ possible ways to join it to to the $W$-shaped component (either telomere of the newly cut crown could join with either telomere of the $W$-shaped component). The result of this operation is one fewer nontrivial crown and an increase in the size of the $W$-shaped component by $c$. These considerations lead to the recursive formulas for $\stg_W(\mathcal{C},w)$.
\end{proof}

We next provide an expression for $\stg_{C}(\mathcal{C})$.

\begin{proposition}\label{prop:recc}
We have the following expression for $\stg_C$:
\[
\stg_C(\mathcal{C})=   
\sum_{c\in \mathcal{C}}\left(c\cdot \stg_W(\mathcal{C}-\st{c},c)\right). 
\]
\end{proposition}

\noindent
Note that $\mathcal{C} = \emptyset$ is not in the domain of $\stg_{C}$, and so this expression is well-defined.

\begin{proof}
When a collection of crowns sorts together, the first operation must be a cut. If the crown being operated on has size $c$, there are $c$ possible places to cut it. The result of this operation is one fewer crown and a $W$-shaped component of size $c$. These considerations lead to the expression for $\stg_C(\mathcal{C})$. 
\end{proof}

\begin{remark}\label{rem:allcrowns}
A closed-form expression for $\stg_{C}(\mathcal{C})$ was previously established in~\cite[Corollary~3.2]{bailey2023complexity}. 
\end{remark}

We have already discussed $\stg_W(\mathcal{C},w)$. Below we establish tools to show $\stg_M(\mathcal{C},w)$ is the same as $\stg_W(\mathcal{C},w)$.
We have thus far considered $A(G_1,G_2)$, where operations act on $V_1$. We can also consider $A(G_2,G_1)$, which is the same as $A(G_1,G_2)$, but we now operate on $V_2$. Every component of $A(G_1,G_2)$ corresponds to a component in $A(G_2,G_1)$ (on the same vertices), though component types are not necessarily the same. For example, an $M$-shaped component $A_{12}$ in $A(G_1,G_2)$ corresponds to a $W$-shaped component $A_{21}$ in $A(G_2,G_1)$.

\begin{definition}
For an operation $\alpha$, define the reverse operation $\alpha^{rev}$ as follows. If $\alpha$ is a cut at adjacency $ab$, then $\alpha^{rev}$ is a join of $a$ and $b$. Similarly, the reverse of a join is a cut. Further, if $\alpha$ is a cut-join $ab,c$ to $ac,b$, then $\alpha^{rev}$ is a cut-join $ac,b$ to $ab,c$. For a sequence $\sigma$ of operations $\sigma_1,\dots,\sigma_k$, let the reverse sequence $\sigma^{rev}$ be $\sigma^{rev}_k,\dots,\sigma^{rev}_1$.  
\end{definition}

Observe that reversing an operation does not change the extremities operated on.
We now show that the reverse operation \textit{rev} preserves the sort together relation.

\begin{proposition} \label{prop:mwsymmetry}
For a sorting scenario $\sigma$ and components $A_{12}$ and $B_{12}$ in $A(G_1,G_2)$, suppose we have that $A_{12}\sim B_{12}$. Then, for $\sigma^{rev}$ and corresponding components $A_{21}$ and $B_{21}$ in $A(G_2,G_1)$, we have $A_{21}\sim B_{21}$. 
\end{proposition}

\begin{proof}
Let $\sigma$ be a sorting scenario for $A(G_1,G_2)$ that consists of a sequence of operations $\sigma_1, \ldots, \sigma_k$. Suppose further that $\sigma$ sorts two components $A_{12}$ and $B_{12}$ together such that $A_{12}\sim B_{12}$. Since $\sigma$ transforms $G_1$ into $G_2$, the sequence $\sigma^{rev}$ transforms $G_2$ into $G_1$. Thus, $\sigma^{rev}$ is a sorting scenario for $A(G_2,G_1)$. 

Let $A_{21}$ and $B_{21}$ denote the two components in $A(G_2,G_1)$ that correspond to $A_{12}$ and $B_{12}$. We will show that $A_{21}\sim B_{21}$ with respect to $\sigma^{rev}$. Since $A_{12}\sim B_{12}$ with respect to $\sigma$, there is some $i$ such that operation $\sigma_i$ is a cut-join operation which joins an extremity $a$ of $A_{12}$ with an extremity $b$ of $B_{12}$. Moreover, $\sigma^{rev}_i$ cuts the adjacency $ab$. Since $a$ and $b$ come from different components, they must have been joined through operation $\sigma^{rev}_j$ for some $j>i$ to create adjacency $ab$ (which is then cut by $\sigma^{rev}_i$). Thus, $A_{12}\sim B_{12}$ in $A(G_2,G_1)$ with respect to $\sigma^{rev}$, as required.
\end{proof}

\begin{corollary}\label{cor:mwsymmetry}
Let $\mathcal{B}_{12}'$ be a subset of the components of $A(G_1,G_2)$, and let $\mathcal{B}_{21}'$ be the corresponding subset of the components of $A(G_2,G_1)$. Then $\stg(\mathcal{B}'_{12})=\stg(\mathcal{B}'_{21})$.  Consequently, $\stg_M(\mathcal{C},m)=\stg_W(\mathcal{C},m)$, and $\stg_{MW}(\mathcal{C},m,w) = \stg_{MW}(\mathcal{C},w,m)$.
\end{corollary}
\begin{proof}
Given a subset $\mathcal{B}'_{12}$ of components of $A(G_1,G_2)$, let $\mathcal{B}'_{21}$ be the set of corresponding components in $A(G_2,G_1)$. Now, let $\Sigma_{12}$ be the set of scenarios that sort the components of $\mathcal{B}'_{12}$ together, and let $\Sigma_{21}$ be the set of scenarios that sort the components of $\mathcal{B}'_{21}$ together. By definition, $\left|\Sigma_{12}\right|=\stg\p{\mathcal{B}'_{12}}$ and $\left|\Sigma_{21}\right|=\stg\p{\mathcal{B}'_{21}}$. Next, let $\Sigma_{12}^{rev}$ and $\Sigma_{21}^{rev}$ be the sets of reverse operations of the scenarios in $\Sigma_{12}$ and $\Sigma_{21}$ respectively. By Proposition~\ref{prop:mwsymmetry}, $\Sigma_{12}^{rev}\subseteq\Sigma_{21}$. Since reverse operations are involutions, we also obtain that $\Sigma_{21}^{rev}\subseteq\Sigma_{12}$. Since $\left|\Sigma_{12}\right|=\left|\Sigma_{12}^{rev}\right|$ and $\left|\Sigma_{21}\right|=\left|\Sigma_{21}^{rev}\right|$, this implies that $\left|\Sigma_{12}\right|\leq\left|\Sigma_{21}\right|$ and $\left|\Sigma_{21}\right|\leq\left|\Sigma_{12}\right|$. Together, these imply that $\left|\Sigma_{12}\right|=\left|\Sigma_{21}\right|$, meaning $\stg(\mathcal{B}'_{12})=\stg(\mathcal{B}'_{21})$, as required. 

We note that an $M$-shaped component in $A(G_{1}, G_{2})$ is a $W$-shaped component in $A(G_{2}, G_{1})$. Thus, we obtain that $\stg_{M}(\mathcal{C},m) = \stg_{W}(\mathcal{C},m)$ and $\stg_{MW}(\mathcal{C},m,w) = \stg_{MW}(\mathcal{C},w,m)$, as desired.
\end{proof}

By Corollary~\ref{cor:mwsymmetry}, $\stg_M(\mathcal{C},m)=\stg_W(\mathcal{C},m)$. However, here we present a different expression for $\stg_M(\mathcal{C},m)$, which will be useful later when we examine $\stg_{MW}$.  To simplify notation, we introduce a new definition. Denote:
\[
\text{sum}(\mathcal{C}) := \sum_{c \in \mathcal{C}} c.
\]

\begin{lemma}\label{lem:recm}
We have the following expression for $\stg_{M}$:

\begin{align*}
\stg_M(\mathcal{C},m)&=\sum_{\eta=0}^{m-1}\sum_{\mathcal{C}'\subseteq \mathcal{C}}\binom{\setsum(\mathcal{C}) +|\mathcal{C}|+m-1}{\setsum(\mathcal{C'}) +|\mathcal{C}'|+\eta}\stg_N(\mathcal{C}',\eta)\cdot \stg_N(\mathcal{C}-\mathcal{C}',m-\eta-1)\\
&+\sum_{c\in \mathcal{C}}\left(c\cdot \stg_{MW}(\mathcal{C}-\st{c},m,c)\right).
\end{align*}
\end{lemma}

\begin{proof}
As we are sorting a collection of crowns with a single $M$ shaped component, we have the following cases for the initial sorting operation. \\ 

\noindent \textbf{Case 1: The initial sorting operation cuts a nontrivial crown.} Fix a nontrivial crown $C$ of size $c \in \mathcal{C}$. This initial cut  transforms $C$ into a $W$-shaped component. There are $c$ places to cut $C$. We now have a (possibly empty) collection $\mathcal{C} - \{c\}$ of nontrivial crowns, a single $M$ of size $m$, and a single $W$ of size $c$. There are $\stg_{MW}(\mathcal{C} - \{c\}, m,c)$ ways of sorting these components together. Summing over all such choices of which initial crown to cut yields the following number of sorting scenarios, as desired:
\[\sum_{c \in \mathcal{C}} c \cdot \stg_{MW}(\mathcal{C} - \{c\}, m,c).\]

\noindent \textbf{Case 2: The initial sorting operation does not cut a nontrivial crown.} We must cut the $M$-shaped component into two $N$-shaped components:  $N_{1}$ of size $\eta$ and $N_{2}$ of size $m-\eta-1$. By Theorem~\ref{thm:invalidsort}(4), $N_{1}$ must sort together with some (possibly empty) subset of nontrivial crowns, and $N_{2}$ must sort together with the remaining crowns.

In order to enumerate the desired sorting scenarios, we follow the strategy of the proof of Lemma~\ref{lem:partcount} found in \cite[Lemma~A.4]{bailey2023complexity}. 
Let $\mathcal{B'}$ be the collection of components after cutting the $M$-shaped component into $N_{1}$ and $N_{2}$, and let $\pi=(\pi_1,\pi_2)$ be a partition of $\mathcal{B}'$ where $\pi_i$ consists of $N_i$ and the crowns that sort together with $N_i$. From Equation (\ref{eq:SubsetDistance}) we have that:  
\[
d(\mathcal{B'}) = \setsum(\mathcal{C}) + |\mathcal{C}| + m - 1.
\] 

Now fix such a partition $\pi = (\pi_{1}, \pi_{2})$ as described above. Let $\mathcal{C}'$ denote the subset of $\mathcal{C}$ corresponding to the sizes of the crowns in $\pi_1$.
Observe the following:
\begin{itemize}
\item $d(\pi_{1}) = \setsum(\mathcal{C'}) + |\mathcal{C}'| + \eta$
\item $d(\pi_{2}) = \setsum(\mathcal{C} - \mathcal{C'}) + |\mathcal{C} - \mathcal{C}'| + m-\eta-1$
\item $\stg(\pi_1)=\stg_N(\mathcal{C}',\eta)$
\item $\stg(\pi_2)=\stg_N(\mathcal{C}-\mathcal{C}',m-\eta-1)$.
\end{itemize}

 Recall that for $i \in \{1,2\}$, $\stg(\pi_i)$ counts the number of ways to sort only $\pi_i$. 
 Since the choice of ordering is independent of the scenario
for each part, we multiply by the binomial coefficient $\binom{d(\mathcal{B}')}{d(\pi_1), d(\pi_2)}$ and sum over all such partitions $(\pi_1, \pi_2)$ to obtain
 \[
\sum_{\pb{\pi_1,\pi_2}}\binom{d(\mathcal{B'})}{d(\pi_{1}), d(\pi_{2})}\stg\p{\pi_1}\stg\p{\pi_2}
\]
sorting scenarios. By the above, this equals
\[
\sum_{\mathcal{C}' \subseteq \mathcal{C}} \binom{\setsum(\mathcal{C}) + |\mathcal{C}| + m - 1}{\setsum(\mathcal{C'}) + |\mathcal{C}'| + \eta} \stg_{N}(\mathcal{C}', \eta) \cdot \stg_{N}(\mathcal{C} - \mathcal{C}', m-\eta-1)
\]
for fixed $\eta$.

Summing over all such $\eta$ yields the total number of sorting scenarios in this case:
\[
\sum_{\eta=0}^{m-1} \sum_{\mathcal{C}' \subseteq \mathcal{C}} \binom{\setsum(\mathcal{C}) + |\mathcal{C}| + m - 1}{\setsum(\mathcal{C'}) + |\mathcal{C}'| + \eta} \stg_{N}(\mathcal{C}', \eta) \cdot \stg_{N}(\mathcal{C} - \mathcal{C}', m-\eta-1).
\]

We add the counts together for Case 1 and Case 2 to obtain the desired number of sorting scenarios.
\end{proof}

It now remains to define a recurrence relation for $\stg_{MW}(\mathcal{C},m,w)$. We begin with the following lemma.

\begin{lemma}\label{lem:recmw}
For $m,w\geq1$, we have the following recurrence relation for $\stg_{MW}$:
\begin{align*} 
\hspace{-0.1in} \stg_{MW}(\mathcal{C},m,w) &= 2\cdot \stg_{MW}(\mathcal{C},m,w-1)+\sum_{c\in \mathcal{C}}\left(4c\cdot \stg_{MW}(\mathcal{C}-\st{c},m,w+c)\right)\\
&\hspace{-0.6in} +\sum_{\eta=0}^{m-1}4\cdot\sum_{\mathcal{C}'\subseteq \mathcal{C}}\binom{\setsum(\mathcal{C})+|\mathcal{C}|+m+w-1}{\setsum(\mathcal{C'}) +|\mathcal{C}'|+\eta}\stg_N(\mathcal{C}',\eta)\cdot \stg_N(\mathcal{C}-\mathcal{C}',m+w-\eta-1). 
\end{align*}
\end{lemma}

\begin{proof}
When sorting a collection of nontrivial crowns together with an $M$-shaped component and a $W$-shaped component, we have the following cases for the initial sorting operation. \\

\noindent \textbf{Case 1: Cut-join one end of the $W$-shaped component.} First, this is impossible if the $W$-shaped component has size $1$. 
Accordingly, notice that $2 \cdot \stg_{MW}(\mathcal{C},m,w-1)=0$ when $w=1$. 
Otherwise, there are $2$ ways to do such a cut-join, one for each end of the $W$-shaped component. The result is a trivial crown (recall that trivial components don't require further sorting) and a decrease in the size of the $W$-shaped component by $1$. This establishes the $2 \cdot \stg_{MW}(\mathcal{C},m,w-1)$ term.

\noindent \\ \textbf{Case 2: Cut-join a nontrivial crown to the $W$-shaped component.} When applying a cut-join operation on a nontrivial crown of size $c$ and $W$-shaped component, there are $c$ possible places to cut the crown, and there are $4$ possible ways to join it to the $W$-shaped component (join either telomere of the newly cut crown to either telomere of the $W$-shaped component). The result is one fewer crown and an increase in the size of the $W$-shaped component by $c$. This establishes the term $\sum_{c\in \mathcal{C}}\left(4c\cdot \stg_{MW}(\mathcal{C}-\st{c},m,w+c)\right)$.

\noindent \\ \textbf{Case 3: Cut-join the $M$-shaped component and $W$-shaped component into two $N$-shaped components.} There are $m$ possible places that we can cut the $M$-shaped component; cutting at the $i^{\text{th}}$ peak from the left end of the path yields an $N$-shaped component of size $i-1$ from the left piece and an $N$-shaped component of size $m-i$ from the right piece. Then there are four possibilities for joining: the $N$-shaped component of size $i-1$ is joined to either of the two telomeres of the $W$-shaped component creating an $N$-shaped component of size $w+i-1$, or the $N$-shaped component of size $m-i$ is joined to either telomere of the $W$-shaped component creating an $N$-shaped component of size $w+m-i$. We will determine the number of sorting scenarios for one of these possibilities, and then show that all four possibilities yield the same number of sorting scenarios.

Consider the first possibility and let $N_1$ be the $N$-shaped component of size $\eta = w+i-1$, and $N_2$ be the $N$-shaped component of size $m+w-\eta-1$. Similarly to the proof in Lemma~\ref{lem:recm}, let $\mathcal{B'}$ be the collection of components in the adjacency graph after this cut-join occurred, and let $\pi = (\pi_1,\pi_2)$ be a partition of $\mathcal{B'}$ where $\pi_i$ consists of $N_i$ and the crowns that sort together with $N_i$. Then from Equation (\ref{eq:SubsetDistance}) we have that:  
\[
d(\mathcal{B'}) = \setsum(\mathcal{C}) + |\mathcal{C}| + m + w - 1.
\]
Fix a partition $\pi =(\pi_1,\pi_2)$, and let $\mathcal{C'}$ denote the sub-multiset of $\mathcal{C}$ corresponding to the sizes of the crowns in $\pi_1$. Then:
\begin{itemize}
\item $d(\pi_{1}) = \setsum(\mathcal{C'}) + |\mathcal{C}'| + \eta$
\item $d(\pi_{2}) = \setsum(\mathcal{C} - \mathcal{C'}) + |\mathcal{C} - \mathcal{C}'| + m+w-\eta-1$
\item $\stg(\pi_1)=\stg_N(\mathcal{C}',\eta)$
\item $\stg(\pi_2)=\stg_N(\mathcal{C}-\mathcal{C}',m+w-\eta-1)$.
\end{itemize}

Recall that for $i \in \{1,2\}$, $\stg(\pi_i)$ counts the number of ways to sort only $\pi_i$. Since the choice of ordering is independent of the scenario
for each part, we multiply by the binomial coefficient $\binom{d(\mathcal{B}')}{d(\pi_1), d(\pi_2)}$ and  sum over all such partitions $(\pi_1, \pi_2)$ to obtain 
 \[
\sum_{\pb{\pi_1,\pi_2}}\binom{d(\mathcal{B'})}{d(\pi_{1}), d(\pi_{2})}\stg\p{\pi_1}\stg\p{\pi_2}
\]
sorting scenarios, or equivalently
\[
\sum_{\mathcal{C}' \subseteq \mathcal{C}} \binom{\setsum(\mathcal{C}) + |\mathcal{C}| + m + w - 1}{\setsum(\mathcal{C'}) + |\mathcal{C}'| + \eta} \stg_{N}(\mathcal{C}', \eta) \cdot \stg_{N}(\mathcal{C} - \mathcal{C}', m + w -\eta - 1)
\]
for fixed $\eta$.

In the first possibility, we considered joining the $N$-shaped component of size $i-1$ to one of the telomeres of the $W$-shaped component. The second possibility, where we join to the other telomere of the $W$-shaped component, would yield the same number of sorting scenarios. Thus, the total number of sorting scenarios for our first two possibilities is:
\begin{align} \label{eq:PossibilityOne}
&2 \cdot \sum_{\mathcal{C}' \subseteq \mathcal{C}} \binom{\setsum(\mathcal{C}) + |\mathcal{C}| + m + w - 1}{\setsum(\mathcal{C'}) + |\mathcal{C}'| + \eta} \stg_{N}(\mathcal{C}', \eta) \cdot \stg_{N}(\mathcal{C} - \mathcal{C}', m + w -\eta - 1).   
\end{align}
Now consider the last two possibilities, where the $N$-shaped component of size $m-i$ is joined to either telomere of the $W$-shaped component. These are symmetric to the first two possibilities, which we can see by interchanging the roles of $i-1$ and $m-i$. So the number of sorting scenarios we obtain from these two possibilities is as in Equation~(\ref{eq:PossibilityOne}). Thus, the total number of sorting scenarios for all four possibilities is:
\[
4 \cdot \sum_{\mathcal{C}' \subseteq \mathcal{C}} \binom{\setsum(\mathcal{C}) + |\mathcal{C}| + m + w - 1}{\setsum(\mathcal{C'}) + |\mathcal{C}'| + \eta} \stg_{N}(\mathcal{C}', \eta) \cdot \stg_{N}(\mathcal{C} - \mathcal{C}', m + w -\eta - 1).
\]

Finally, we sum over all possible sizes of the initial two $N$-shaped components formed from the cut-join of the $M$-shaped component to the $W$-shaped component (i.e. we sum over all possible cutting locations), which gives the total number of sorting scenarios in this case:
\[
\sum_{\eta = 0}^{m-1} 4\cdot\sum_{\mathcal{C}' \subseteq \mathcal{C}} \binom{\setsum(\mathcal{C}) + |\mathcal{C}| + m + w - 1}{\setsum(\mathcal{C'}) + |\mathcal{C}'| + \eta} \stg_{N}(\mathcal{C}', \eta) \cdot \stg_{N}(\mathcal{C} - \mathcal{C}', m + w -\eta - 1).
\]
This gives the desired result.  
\end{proof}

A priori, to use Lemma~\ref{lem:recmw}, we have to track the case when an $M$-shaped component sorting with a $W$-shaped component results in two $N$-shaped components. This requires $O(2^{|\mathcal{C}|})$ steps to compute the double-summation in Lemma~\ref{lem:recmw}. The next proposition allows us to both avoid these steps and simplify our algorithm.

\begin{proposition}\label{prop:recmw}
For $m,w\geq1$, we have the following recurrence relation for $\stg_{MW}$:
\begin{align*}
\stg_{MW}(\mathcal{C},m,w)&=\stg_{MW}(\mathcal{C},m-1,w)+\stg_{MW}(\mathcal{C},m,w-1)+2\cdot \stg_{W}(\mathcal{C},m+w)\\
&\hspace{-0.8in}+\sum_{c\in \mathcal{C}}2c \,\biggr(\stg_{MW}(\mathcal{C}-\st{c},m+c,w) + 
\stg_{MW}(\mathcal{C}-\st{c},m,w+c) - \stg_{MW}(\mathcal{C}-\st{c},m+w,c)\biggr).
\end{align*} 
\end{proposition}

\begin{proof}

Let $m,w\ge 1$. Applying Corollary~\ref{cor:mwsymmetry} and Lemma~\ref{lem:recmw} below gives
\begin{align}\label{case4eqn}
\stg_{MW}(\mathcal{C},m,w)&=\frac{1}{2}\pb{\stg_{MW}(\mathcal{C},m,w)+\stg_{MW}(\mathcal{C},w,m)}\notag\\
&\hspace{-1in}=\stg_{MW}(\mathcal{C},m,w-1)+\sum\limits_{c\in\mathcal{C}}\left(2c\cdot \stg_{MW}(\mathcal{C}-\st{c},m,w+c)\right)\notag\\
&\hspace{-0.9in}+\sum_{\eta=0}^{m-1}2\cdot\sum_{\mathcal{C}'\subseteq \mathcal{C}}\binom{\setsum(\mathcal{C})+\ab{\mathcal{C}}+m+w-1}{\setsum(\mathcal{C'})+\ab{\mathcal{C}'}+\eta}\stg_N(\mathcal{C}',\eta)\cdot \stg_N(\mathcal{C}-\mathcal{C}',m+w-\eta-1)\notag\\
&\hspace{-0.9in}+\stg_{MW}(\mathcal{C},w,m-1)+\sum\limits_{c\in\mathcal{C}}\left(2c\cdot \stg_{MW}(\mathcal{C}-\st{c},w,m+c)\right)\notag\\
&\hspace{-0.9in}+\sum_{\eta=0}^{w-1}2\cdot\sum_{\mathcal{C}'\subseteq \mathcal{C}}\binom{\setsum(\mathcal{C})+\ab{\mathcal{C}}+w+m-1}{\setsum(\mathcal{C'})+\ab{\mathcal{C}'}+\eta}\stg_N(\mathcal{C}',\eta)\cdot \stg_N(\mathcal{C}-\mathcal{C}',w+m-\eta-1).
\end{align}

\noindent As we are summing over all subsets of $\mathcal{C}$, we can replace $\mathcal{C}'$ with $\mathcal{C}-\mathcal{C}'$ in the last summation of Equation (\ref{case4eqn}), yielding
\begin{align*}
    \sum_{\eta=0}^{w-1}2\cdot\sum_{\mathcal{C}'\subseteq \mathcal{C}}\binom{\setsum(\mathcal{C})+\ab{\mathcal{C}}+w+m-1}{\setsum(\mathcal{C}-\mathcal{C'})+\ab{\mathcal{C}-\mathcal{C}'}+\eta}\stg_N(\mathcal{C}-\mathcal{C}',\eta)\cdot\stg_N(\mathcal{C}',w+m-\eta-1).
\end{align*}

\noindent Using the identity $\binom{n}{k}=\binom{n}{n-k}$ and rearranging the order of multiplication, this summation becomes
\begin{align*}
    \hspace{-0.1in}&\sum_{\eta=0}^{w-1}2\cdot\sum_{\mathcal{C}'\subseteq \mathcal{C}}\binom{\setsum(\mathcal{C})+\ab{\mathcal{C}}+w+m-1}{\setsum(\mathcal{C'})+\ab{\mathcal{C}'}+w+m-\eta-1}\stg_N(\mathcal{C}',w+m-\eta-1)\cdot \stg_N(\mathcal{C}-\mathcal{C}',\eta).
\end{align*}

\noindent Note that we can reverse the order of this summation by replacing $w+m-\eta-1$ with $\eta$ and summing over $\eta$ from $m$ to $m+w-1$. This re-indexing changes the previous summation to
\begin{align*}
    \sum_{\eta=m}^{m+w-1}2\cdot\sum_{\mathcal{C}'\subseteq \mathcal{C}}\binom{\setsum(\mathcal{C})+\ab{\mathcal{C}}+m+w-1}{\setsum(\mathcal{C'})+\ab{\mathcal{C}'}+\eta}\stg_N(\mathcal{C}',\eta)\cdot \stg_N(\mathcal{C}-\mathcal{C}',m+w-\eta-1).
\end{align*}

\noindent Substituting this sum back into the last summation of equation (\ref{case4eqn}) and combining summations over $\eta$, we have
\begin{align}\label{case4eqn3}
\hspace{-0.1in}\stg_{MW}(\mathcal{C},m,w)&=\stg_{MW}(\mathcal{C},m,w-1)+\sum\limits_{c\in\mathcal{C}}\left(2c\cdot \stg_{MW}(\mathcal{C}-\st{c},m,w+c)\right)\notag\\
&\hspace{-1.0in}+\sum_{\eta=0}^{m-1}2\cdot\sum_{\mathcal{C}'\subseteq \mathcal{C}}\binom{\setsum(\mathcal{C})+\ab{\mathcal{C}}+m+w-1}{\setsum(\mathcal{C'})+\ab{\mathcal{C}'}+\eta}\stg_N(\mathcal{C}',\eta)\cdot \stg_N(\mathcal{C}-\mathcal{C}',m+w-\eta-1)\notag\\
&\hspace{-1.0in}+\stg_{MW}(\mathcal{C},w,m-1)+\sum\limits_{c\in\mathcal{C}}\left(2c\cdot \stg_{MW}(\mathcal{C}-\st{c},w,m+c)\right)\notag\\
&\hspace{-1.0in}+\sum_{\eta=m}^{m+w-1}2\cdot\sum_{\mathcal{C}'\subseteq \mathcal{C}}\binom{\setsum(\mathcal{C})+\ab{\mathcal{C}}+m+w-1}{\setsum(\mathcal{C'})+\ab{\mathcal{C}'}+\eta}\stg_N(\mathcal{C}',\eta)\cdot \stg_N(\mathcal{C}-\mathcal{C}',m+w-\eta-1)\notag\\
&\hspace{-1.1in}=\stg_{MW}(\mathcal{C},m,w-1)+\sum\limits_{c\in\mathcal{C}}\left(2c\cdot \stg_{MW}(\mathcal{C}-\st{c},m,w+c)\right)\notag\\
&\hspace{-1.0in}+\stg_{MW}(\mathcal{C},w,m-1)+\sum\limits_{c\in\mathcal{C}}\left(2c\cdot \stg_{MW}(\mathcal{C}-\st{c},w,m+c)\right)\notag\\
&\hspace{-1.0in}+\sum_{\eta=0}^{m+w-1}2\cdot\sum_{\mathcal{C}'\subseteq \mathcal{C}}\binom{\setsum(\mathcal{C})+\ab{\mathcal{C}}+m+w-1}{\setsum(\mathcal{C'})+\ab{\mathcal{C}'}+\eta}\stg_N(\mathcal{C}',\eta)\cdot \stg_N(\mathcal{C}-\mathcal{C}',m+w-\eta-1).
\end{align}
By rearranging the equation in Lemma~\ref{lem:recm} and multiplying it by 2, note that the last summation in (\ref{case4eqn3}) can be rewritten as
\begin{align*}
&\sum_{\eta=0}^{m+w-1}2\cdot\sum_{\mathcal{C}'\subseteq \mathcal{C}}\binom{\setsum(\mathcal{C})+\ab{\mathcal{C}}+m+w-1}{\setsum(\mathcal{C'})+\ab{\mathcal{C}'}+\eta}\stg_N(\mathcal{C}',\eta)\cdot \stg_N(\mathcal{C}-\mathcal{C}',m+w-\eta-1)\\
&\hspace{0.4in}=2\stg_M(\mathcal{C},m+w) -\sum\limits_{c\in\mathcal{C}}\left(2c\cdot \stg_{MW}(\mathcal{C}-\st{c},m+w,c)\right)\\
&\hspace{0.4in}=2\stg_W(\mathcal{C},m+w)-\sum\limits_{c\in\mathcal{C}}\left(2c\cdot \stg_{MW}(\mathcal{C}-\st{c},m+w,c)\right),
\end{align*}
where the second equality comes from Corollary~\ref{cor:mwsymmetry}.
Substituting this back into equation (\ref{case4eqn3}), rearranging terms, and noting that $\stg_{MW}(\mathcal{C},w,m-1)=\stg_{MW}(\mathcal{C},m-1,w)$ and $\stg_{MW}(\mathcal{C}-\st{c},w,m+c)=\stg_{MW}(\mathcal{C}-\st{c},m+c,w)$ by Corollary~\ref{cor:mwsymmetry}, yields the desired result as follows:
\begin{align*}
\stg_{MW}(\mathcal{C},m,w)=&\stg_{MW}(\mathcal{C},m,w-1)+\sum\limits_{c\in\mathcal{C}}\left(2c\cdot \stg_{MW}(\mathcal{C}-\st{c},m,w+c)\right)\\
&+\stg_{MW}(\mathcal{C},w,m-1)+\sum\limits_{c\in\mathcal{C}}\left(2c\cdot \stg_{MW}(\mathcal{C}-\st{c},w,m+c)\right)\\
&+2\stg_W(\mathcal{C},m+w)-\sum\limits_{c\in\mathcal{C}}\left(2c\cdot \stg_{MW}(\mathcal{C}-\st{c},m+w,c)\right)\\
=&\stg_{MW}(\mathcal{C},m-1,w)+\stg_{MW}(\mathcal{C},m,w-1)+ 2\cdot \stg_{W}(\mathcal{C},m+w)\\
&+\sum\limits_{c\in\mathcal{C}}2c\,\biggr(\stg_{MW}(\mathcal{C}-\st{c},m+c,w) + 
\stg_{MW}(\mathcal{C}-\st{c},m,w+c)\\ &- \stg_{MW}(\mathcal{C}-\st{c},m+w,c)\biggr).
\end{align*}
\end{proof}

We have the following corollary, which shows that the numbers $\stg_{MW}(\emptyset,m,w)$ satisfy a Pascal-like recurrence.
\begin{corollary}\label{cor:recmw}
    For all $m,w\geq1$,
    \[
    \stg_{MW}(\emptyset,m,w)=\stg_{MW}(\emptyset,m-1,w)+\stg_{MW}(\emptyset,m,w-1)+2^{m+w}.
    \]
\end{corollary}

\begin{proof}
If $\mathcal{C}=\emptyset$, the summation in Proposition~\ref{prop:recmw} is empty. So,
\[
\stg_{MW}(\emptyset,m,w)=\stg_{MW}(\emptyset,m-1,w)+\stg_{MW}(\emptyset,m,w-1)+2\cdot \stg_W(\emptyset,m+w).
\]
By Lemma~\ref{lem:SortOneComponent}, $\stg_W(\emptyset,m+w)=2^{m+w-1}$. The result now follows.
\end{proof}

We arrived at Corollary~\ref{cor:recmw} after our initial conference paper~\cite{BaileySWAT}. As a result of the insights gleaned, we were able to greatly simplify the statements of Lemma~\ref{lem:recmw} and Proposition~\ref{prop:recmw}, and the proof of Proposition~\ref{prop:recmw}.

\section{Fixed-Parameter Tractability of \algprobm{Pairwise Rearrangement}} \label{sec:FPT}

In this section, we will use the recurrences we have found to  establish our main result. 

\begin{thm:MainFPT}\label{thm:fpt}
In the Single Cut-and-Join model, \algprobm{Pairwise Rearrangement} is fixed-parameter tractable, with respect to the number of nontrivial components in the adjacency graph.
\end{thm:MainFPT}

\begin{proof}
Given two genomes with $n$ edges each, the adjacency graph is easily constructed in polynomial time. Let $\mathcal{B}$ be the set of components for this adjacency graph, and let $k$ denote the number of nontrivial components in $\mathcal{B}$. It suffices to show that the number of scenarios $\mps(\mathcal{B})$ can be determined in time $O(k\cdot 2^k \cdot B_k+k \cdot 2^{k} \cdot n^2)$, where $B_k$ is the $k^{\text{th}}$ Bell number. 

Let $\mathcal{C}, \mathcal{M}, \mathcal{W}, \mathcal{N}$ be multisets of nonnegative integers, denoting the sizes of the nontrivial crowns, $M$-shaped components, $W$-shaped components, and $N$-shaped components, respectively, of $\mathcal{B}$. Thus, we assume that every element of $\mathcal{C}$ is at least 2, and every element of $\mathcal{M}$ and $\mathcal{W}$ is at least 1. 

We will proceed in two stages. In the first stage, we will build up a table of values for $\stg_N(\mathcal{A},\eta)$, $\stg_W(\mathcal{A},w)$, $\stg_C(\mathcal{A})$, and $\stg_{MW}(\mathcal{A},m,w)$ with $\mathcal{A}$ any sub-multiset of $\mathcal{C}$, along with 
$\eta\geq0$, $w\geq1$, and $m\geq1$. Then in the second stage, we will use this final table of values along with Lemma~\ref{lem:partcount} to compute $\mps(\mathcal{C},\mathcal{M},\mathcal{W},\mathcal{N})$.

\noindent \\ \textbf{Stage 1.} When $\eta \geq 0$ we have by Lemma~\ref{lem:SortOneComponent} that $\stg_{N}(\emptyset, \eta) = 1$. Similarly, for $w \geq 1$, we have that $\stg_{W}(\emptyset, w) = 2^{w-1}$, which is computable in time $O(n).$ 

Now fix $m \in [n], w \in [n]$ with $m+w=i$, and suppose that for each $1 \leq m' \leq m$ and each $1 \leq w' \leq w$ with $m'+w'<i$, that we have computed $\stg_{MW}(\emptyset, m',w')$. Using these values, together with the fact that $\stg_{W}(\emptyset, w) = 2^{w-1}$, we may apply Corollary~\ref{cor:recmw} to compute $\stg_{MW}(\emptyset, m, w)$. There are at most $n^2$ pairs $(m',w') \in [n] \times [n]$, and so this step takes time $O(n^2)$.

Let $u$ be an integer $0<u \leq |\mathcal{C}|\leq k$. First, let $0<x\leq n$.
Suppose we have a table of values of $\stg_N(\mathcal{A},\eta)$ and $\stg_W(\mathcal{A},w)$ for all sub-multisets $\mathcal{A}$ of $\mathcal{C}$ with $|\mathcal{A}|<u$ and $\eta,w\leq n$ and also for all sub-multisets $\mathcal{A}$ of $\mathcal{C}$ with $|\mathcal{A}|=u$ and $\eta,w<x$. 
Let $\mathcal{C}'$ be a sub-multiset of $\mathcal{C}$ with $|\mathcal{C}'|=u$. 
We now proceed to fill in some of our lookup table from the bottom-up, using the recurrences from the previous section.

\begin{itemize}
    \item We may compute $\stg_{N}(\mathcal{C'},x)$ using Proposition~\ref{prop:recn}. There are at most $kn$ cells in our lookup table, and so this step takes time $O(kn)$. 

    \item We may compute $\stg_{W}(\mathcal{C'}, x)$  using Proposition~\ref{prop:recw}. There are at most $kn$ cells in our lookup table, and so this step takes time $O(kn)$.

    \item We now turn to computing $\stg_{C}(\mathcal{C'})$. To do so, we apply Proposition~\ref{prop:recc} using the previously computed values of $\stg_{W}(\mathcal{A},w)$ for all $\mathcal{A}$ of size less than $u$ and all $1 \leq w \leq n$. This step takes time $O(k)$. 
\end{itemize}

Now, let $x$ and $y$ be nonnegative integers with $x+y\leq n$. Suppose we have a table of values of $\stg_{MW}(\mathcal{A},m,w)$ for all sub-multisets $\mathcal{A}$ of $\mathcal{C}$ satisfying one of the following:
\begin{itemize}
    \item $|\mathcal{A}|<u$ for any nonnegative integers $m$ and $w$ with $m+w\leq n$, or 
    \item $|\mathcal{A}|=u$ with $m+w<x+y$.
\end{itemize}
We now turn to computing $\stg_{MW}(\mathcal{C'}, x, y)$. 
To do so, we apply Proposition~\ref{prop:recmw} using the previously computed values of $\stg_{W}(\mathcal{A},w)$ for all $1 \leq w \leq n$ and of $\stg_{MW}(\mathcal{A},m,w)$ for all sub-multisets $\mathcal{A}$ of $\mathcal{C}$ satisfying one of the conditions above.
There are at most $k n^2$ many such cells in our lookup table, and so filling in these cells takes time $O(k n^2)$.

To summarize, our first stage involves building up a table of values for $\stg_N(\mathcal{A},\eta)$, $\stg_W(\mathcal{A},w)$, $\stg_C(\mathcal{A})$, and $\stg_{MW}(\mathcal{A},m,w)$ with $\mathcal{A}$ any sub-multiset of $\mathcal{C}$ and $\eta$, $w$, and $m$ having values as above. This first stage is computable in time $O(k \cdot2^{k} \cdot n^{2})$ by iterating through all sub-multisets of $\mathcal{C}$ in non-decreasing order of cardinality.

\noindent \\ \textbf{Stage 2.} We will use the table of values computed in Stage 1, along with Lemma~\ref{lem:partcount}, to compute $\mps(\mathcal{C},\mathcal{M},\mathcal{W},\mathcal{N})$ as follows. First, let $\mathcal{B}' \subseteq \mathcal{B}$ denote the set of all nontrivial components in $\mathcal{B}$, and let $t$ denote the number of trivial $N$-shaped components in $\mathcal{B}$. As per Lemma~\ref{lem:partcount}, we consider each partition $\pi$ of $\mathcal{B}$, and examine $\stg(\pi_i)$ for each part $\pi_i$ of $\pi$. By Theorem~\ref{thm:invalidsort}, we can determine in time $O(k)$ if $\stg(\pi_{i}) = 0$. In such a case, the term of $\mps(\mathcal{B})$ corresponding to $\pi$ contributes $0$, and so we may discard $\pi$. Call a partition $\pi$ \emph{permissible} if $\stg(\pi_i)>0$ for all parts. Recall the parts $\pi_i$ from Theorem~\ref{thm:invalidsort} for which $\stg(\pi_i)=0$. Note that the only part $\pi_i$ with a combination of trivial and nontrivial components for which $\stg(\pi_i)>0$ is the case in which $\pi_{i}$ contains precisely nontrivial crowns and a single trivial $N$-shaped component. 
So the permissible partitions of $\mathcal{B}$ are precisely the partitions generated by the following procedure: 
\begin{itemize}
\item Start with a permissible partition $\pi'$ of $\mathcal{B}'$.
\item For at most $t$ parts $\pi_i'$ of $\pi'$ that each consist of only nontrivial crowns, add one of the $t$ trivial $N$-shaped components from $\mathcal{B}$ to $\pi_i'$. 

\item Each of the remaining trivial $N$-shaped components and each trivial crown from $\mathcal{B}$ become a singleton part. 
\end{itemize}

Let $\pi' = (\pi_{1}', \ldots, \pi_{p(\pi')}')$ be a permissible partition of $\mathcal{B}'$. We say that a permissible partition $\pi = (\pi_{1}, \ldots, \pi_{p(\pi)})$ of $\mathcal{B}$ \emph{extends} $\pi'$ if there exists an injective function $f : [p(\pi')] \to [p(\pi)]$ such that $\pi_{i}' \subseteq \pi_{f(i)}$ for all $i \in [p(\pi')]$. Note that by Theorem~\ref{thm:invalidsort}, $\pi_{i}'$ and $\pi_{f(i)}$ may only be different if $\pi_{i}'$ contains only nontrivial crowns, and $\pi_{f(i)}$ contains a nontrivial $N$-shaped component.

Note that while $\mathcal{B}'$ has at most $k$ components, $\mathcal{B}$ may in general not have bounded size. Fix a permissible partition $\pi'$ of $\mathcal{B}'$. Our goal is to compute:
\[
\sum_{\substack{\pi \in \Pi(\mathcal{B}) \\ \pi \text{ extends } \pi'}} \binom{d(\mathcal{B})}{d(\pi_{1}), d(\pi_{2}), \ldots, d(\pi_{p(\pi)})} \prod_{i=1}^{p(\pi)} \stg(\pi_{i}).
\]
As above, let $\pi'$ be a permissible partition of $\mathcal{B}'$, and let $\pi$ be an extension of $\pi'$. Recall, if $\pi_{i}$ consists of only one trivial component, then $\stg(\pi_{i}) = 1$. 
Let $\pi_{i}$ be a part of $\pi$ that contains only nontrivial crowns and a single trivial $N$-shaped component. Let $\pi_{i}'$ be obtained from $\pi_{i}$ by removing the trivial $N$-shaped component. We have by equation~(\ref{eq:SubsetDistance}) that $d(\pi_{i}') = d(\pi_{i})$.  Thus, if $\pi$ extends $\pi'$, we have that:
\[
\binom{d(\mathcal{B})}{d(\pi_{1}), d(\pi_{2}), \ldots, d(\pi_{p(\pi)})} = \binom{d(\mathcal{B}')}{d(\pi_{1}'), d(\pi_{2}'), \ldots, d\!\pb{\pi_{p(\pi')}'}}.
\]

Theorem~\ref{thm:invalidsort} and the above procedure for constructing permissible partitions yields precisely the following cases:

\begin{itemize}
\item \textbf{Case 1:} Suppose that $\pi_{i}'$ consists of an $N$-shaped component and zero or more nontrivial crowns. In this case, we can use a known value of $\stg_N$ to count the number of ways to sort this part.

\item \textbf{Case 2:} Suppose that $\pi_{i}'$  consists of a $W$-shaped component and zero or more nontrivial crowns. In this case, we can use a known value of $\stg_W$ to count the number of ways to sort this part.

\item \textbf{Case 3:} Suppose that  $\pi_{i}'$ consists of an $M$-shaped component and zero or more nontrivial crowns. In this case, we can also use a known value of $\stg_M$ to count the number of ways to sort this part.

\item \textbf{Case 4:} Suppose that $\pi_{i}'$ consists of an $M$-shaped component, a $W$-shaped component, and zero or more nontrivial crowns. In this case, we can use a known value of $\stg_{MW}$ to count the number of ways to sort this part.

\item \textbf{Case 5:}  If none of Cases 1-4 hold, then we necessarily have that $\pi_{i}'$ consists of only nontrivial crowns. Since the presence of such parts implies that $\pi'$ may not necessarily extend uniquely to a permissible partition of $\mathcal{B}$, we consider all such parts $\pi_i'$ simultaneously in order to handle all such permissible extensions of $\pi'$. Let $\pi_{i_{1}}', \ldots, \pi_{i_{\ell}}'$ be the parts of $\pi'$ that contain only nontrivial crowns. By Theorem~\ref{thm:invalidsort}, we have that for $j \in [\ell]$, we can only add at most one trivial $N$-shaped component to $\pi_{i_{j}}'$. As the trivial $N$-shaped components and the parts of $\pi'$ are both distinguishable, there are $P(t,v) = t! / (t-v)!$ ways to assign $v \leq t$ trivial $N$-shaped components to $v$ parts drawn from $\pi_{i_{1}}', \ldots, \pi_{i_{\ell}}'$. The remaining trivial $N$-shaped components each form their own singleton component in the corresponding partition $\pi$ of $\mathcal{B}$ extending $\pi'$. If part $\pi_{i_{j}}'$ ($j \in [\ell]$) has a trivial $N$-shaped component added, we can use a known value of $\stg_N$ to count the number of ways to sort this part. If not, we can use a known value of $\stg_C$ to count the number of ways to sort this part. The number of ways to sort each part is independent of the number of ways to assign the trivial $N$-shaped components. Thus, multiply the value returned from the lookup table by $P(t,v)$. Finally, using Lemma~\ref{lem:partcount}, we multiply together each of the $\stg_{N}$ and $\stg_{C}$ values that we accessed from the lookup table. This takes time $O(k)$ since there are $O(k)$ parts.

\noindent We now iterate over all integers $0\leq v\leq\min(t, \ell)$, where we consider all possible assignments of $v$ trivial $N$-shaped components to the parts of $\pi'$ (where the assignment is as described in the preceding paragraph). As each $\pi_{i_{j}}'$ can only receive at most one trivial $N$-shaped component, we iterate over all $v$-subsets of $[\ell]$ to fully enumerate each case. As $\ell \leq k$, there are at most $2^{k}$ total subsets to consider across all values of $v$, and so computing the count in this case takes time $O(k \cdot 2^{k})$.

\end{itemize}

In the first four cases, we only look at a single part at a time. Since there are $O(k)$ parts, Lemma~\ref{lem:partcount} gives a complexity of $O(k)$ for each of these cases. In Case~5 we already account for Lemma~\ref{lem:partcount}, so we have a complexity of $O(k\cdot 2^k$). Thus for all cases the runtime is at most $O(k\cdot 2^k)$ using the existing values in the lookup table to compute $\stg(\pi_i)$. As we are summing over all possible partitions, we are evaluating at most the number of partitions of a set of cardinality $k$, which is the $k^{\text{th}}$ Bell number $B_k$. This gives a complexity of $O(k\cdot 2^k\cdot B_k)$ for computing $\mps(\mathcal{C},\mathcal{M},\mathcal{W},\mathcal{N})$ from the pre-established lookup tables.

So, overall, we have a complexity of $O( k \cdot 2^{k} \cdot n^2 )$ for our first stage and a complexity of $O(k\cdot 2^k\cdot B_k)$ for our second stage. Thus, $\mps(\mathcal{C},\mathcal{M},\mathcal{W},\mathcal{N})$ is computable in time $O(k \cdot 2^k \cdot B_k+k \cdot 2^{k} \cdot n^2)$, as required. The result now follows. 
\end{proof}

\section{Uniform Sampling is Fixed-Parameter Tractable}

In this section, we exhibit a fixed-parameter tractable algorithm to sample sorting scenarios.

\begin{corollary} \label{cor:Sampling}
Let $G_1, G_2$ be genomes with $n$ genes, and let $A(G_1, G_2)$ be the adjacency graph (see Definition~\ref{def:adjgraph}). Let $k$ be the number of nontrivial components in $A(G_1, G_2)$. Then the problem of uniformly sampling  sorting scenarios transforming $G_1$ into $G_2$, is fixed-parameter tractable with respect to $k$. 
\end{corollary}

Recall that if $d(G_1, G_2)$ is bounded, then so is the number of nontrivial components in the adjacency graph $A(G_1, G_2)$. So in light of Corollary~\ref{cor:Sampling}, we obtain a fixed-parameter tractable algorithm to sample sorting scenarios, with respect to distance.

\begin{proof}
Let $d := d(G_1, G_2)$. Let $S = (s_1, \ldots, s_{d})$ be a sorting scenario transforming $G_1$ into $G_2$, and let $S' = (s_1, \ldots, s_{d'})$ for some $d' \leq d$. Let $G_{1}'$ be the genome obtained by applying $S'$ to $G_1$. As there are at most $k$ nontrivial  components in $A(G_1, G_2)$ and $k$ is bounded, we have by Observation~\ref{obs:SortingScenarios} (see Figure~\ref{fig:sortingops}) that for any such $S'$, the number of nontrivial components in $A(G_{1}', G_2)$ is at most $2^k$.

We construct a uniform sampler that runs in time $f(k) \cdot \poly(n)$, for some function $f(k)$. We proceed as follows. For each sorting operation $t$, define $G_{1,t}$ to be the genome obtained by applying $t$ to $G_1$. We choose our initial sorting operation with probability $\mps(G_{1,t}, G_{2})/\mps(G_{1}, G_{2})$. By Theorem~\ref{thm:MainFPT}, we can solve \algprobm{Pairwise Rearrangement} for $G_{1,t}$ and $G_2$ in time $f(k) \cdot \poly(n)$, for some function $f(k)$. Genome $G_1$ has $O(n)$ adjacencies and $O(n)$ telomeres, so there are $O(n)$ possible cuts, $O(n)$ possible joins, and $O(n^2)$ possible cut-joins. So, there are $O(n^2)$ \emph{total} operations available, so $O(n^2)$ is also a bound on the number of sorting operations. Thus, in time $f(k) \cdot \poly(n)$, we can compute each $\mps(G_{1,t}, G_{2})/\mps(G_{1}, G_{2})$.

Now suppose we pick sorting operation $t'$ as our initial sorting operation. We now recurse, starting from $G_{1,t'}$ instead of $G_{1}$. Observe that $d(G_1, G_2) \leq 2n$; a genome with $n$ genes can have at most $n$ adjacencies (each gene has 2 extremities and the total $2n$ extremities could be paired into $n$ adjacencies). For distance, we could first cut all adjacencies that are in $G_1$ but not $G_2$. Then we can use joins to obtain all adjacencies which are in $G_2$ but not $G_1$. In particular, $d(G_1, G_2)$ is polynomial in $n$, and so our recursion depth is $\poly(n)$. The result now follows.
\end{proof}

\section{Conclusion}

We investigated the computational complexity of the \algprobm{Pairwise Rearrangement} problem in the Single Cut-and-Join model. In particular, while \algprobm{Pairwise Rearrangement} was previously shown to be $\textsf{\#P}$-complete under polynomial-time Turing reductions~\cite{bailey2023complexity}, we proved it is fixed-parameter tractable when parameterized by the number of nontrivial components in the adjacency graph (Theorem~\ref{thm:MainFPT}). In particular, our results show that the number of nontrivial components serves as a key barrier towards efficiently enumerating and sampling minimum length sorting scenarios. We also obtain a fixed-parameter tractable algorithm to uniformly sample sorting scenarios, with respect to the number of nontrivial components in the adjacency graph. We conclude with some open questions.

\begin{question}
In the Single Cut-and-Join model, does \algprobm{Pairwise Rearrangement} belong to \textsf{FPRAS}?    
\end{question}

\begin{question}
In the Double Cut-and-Join model, is \algprobm{Pairwise Rearrangement} $\textsf{\#P}$-complete?
\end{question}

The work of Braga and Stoye~\cite{BragaStoyeDCJ} immediately yields a fixed-parameter tractable algorithm for \algprobm{Pairwise Rearrangement} in the Double Cut-and-Join model. Their algorithm is considerably simpler than our work in this paper.

The computational complexity of the \algprobm{Pairwise Rearrangement} problem in the Reversal model is a long-standing open question. In particular, it is believed that this problem is $\textsf{\#P}$-complete. Furthermore, \algprobm{Pairwise Rearrangement} has been resistant to efficient sampling; membership in $\textsf{FPRAS}$ remains open. Thus, perhaps the lens of parameterized complexity might shed new light on this problem. The following question is a natural place to start. 

\begin{question}
In the Reversal model, is \algprobm{Pairwise Rearrangement} fixed-parameter tractable by the number of components in the adjacency graph?
\end{question}

\bibliographystyle{plainurl}
\bibliography{references}
\end{document}